\documentclass[twocolumn,natbib]{svjour3}          

\smartqed

\usepackage{times}
\usepackage[T1]{fontenc}
\usepackage{color, colortbl}
\usepackage{graphicx,color,psfrag}

\definecolor{lightgray}{rgb}{0.9,0.9,0.9}
\newcolumntype{g}{>{\columncolor{lightgray}}c}

\usepackage{graphics,graphicx}
\usepackage{esvect}
\usepackage{amsmath,amssymb}
\usepackage{url}

\usepackage{caption}
\usepackage{subcaption}

\usepackage{thmtools, thm-restate}
\renewcommand\thmcontinues[1]{Continued}

\usepackage[nomargin]{fixme}
\usepackage[breaklinks,hidelinks]{hyperref}  
\usepackage[hyphenbreaks]{breakurl}

\fxuselayouts{pdfcsigmargin}
\fxsetup{layout={inline,pdfcsigmargin}}

\fxuseenvlayout{signature}
\fxloadtargetlayouts{changebar}
\fxusetargetlayout{changebar}
\FXRegisterAuthor{fc}{afc}{FC}

\declaretheorem[style=definition,name=Example]{ex}

\newcommand{\eg}{e.g.}
\newcommand{\ie}{i.e.}
\newcommand{\viz}{viz.}

\newcommand{\suchthat}{\ensuremath{\mbox{s.t.}}}

\newcommand{\T}[1]{\ensuremath{T_{#1}}}

\newcommand{\C}[1]{\ensuremath{C_{#1}}}

\newcommand{\W}[1]{\ensuremath{W_{#1}}}
\newcommand{\m}{\ensuremath{m}}

\newcommand{\realset}{\ensuremath{\mathbb{R}}}
\newcommand{\opinionset}{\ensuremath{\mathbb{O}}}
\newcommand{\admissibleopinionset}[1]{\ensuremath{\mathbb{O}_{#1}}}

\newcommand{\tuple}[1]{\ensuremath{\langle #1 \rangle}}
\newcommand{\belief}[1]{\ensuremath{b_{#1}}}
\newcommand{\disbelief}[1]{\ensuremath{d_{#1}}}
\newcommand{\uncertainty}[1]{\ensuremath{u_{#1}}}
\newcommand{\opinion}[1]{\ensuremath{\tuple{\belief{#1},\allowbreak{} ~\disbelief{#1},\allowbreak{} ~\uncertainty{#1}}}}

\newcommand{\distance}[2]{\ensuremath{d_G(#1, #2)}}
\newcommand{\distanceexp}[2]{\ensuremath{d_E(#1, #2)}}

\newcommand{\optrustconf}[3]{\ensuremath{#1 \circ_{#3} #2}}
\newcommand{\optrustconffam}[3]{\ensuremath{#1 ~\overline{\circ_{\alpha_{C'}}}~ #2}}

\newcommand{\tri}[1]{\ensuremath{\overset{\triangle}{#1}}}
\newcommand{\ang}[1]{\ensuremath{\angle_{#1}}}

\newcommand{\versor}[1]{\ensuremath{\vv{e_{#1}}}}

\newcommand{\fusionop}[1]{\ensuremath{\Gamma_{#1}}}

\newcommand{\set}[1]{\ensuremath{\{#1\}}}
\newcommand{\setagents}{\ensuremath{\mathcal{A}}}
\newcommand{\agent}[1]{\ensuremath{a_{#1}}}
\newcommand{\neighbours}[1]{\ensuremath{N_{#1}}}

\newcommand{\opinionag}[2]{\ensuremath{O_{#1}^{#2}}}
\newcommand{\truthprob}[1]{\ensuremath{P_{#1}^{T}}}

\newcommand{\linkprob}{\ensuremath{P^L}}

\newcommand{\setbeliefs}[1]{\ensuremath{\mathcal{KB}_{#1}}}

\newcommand{\sharedbelief}{\ensuremath{\Omega}}

\newcommand{\truesymbol}{\ensuremath{\top}}
\newcommand{\falsesymbol}{\ensuremath{\perp}}

\newcommand{\lboot}{\ensuremath{\#_{B}}}

\newcommand{\josangdiscountsym}{\ensuremath{\otimes}}
\newcommand{\josangfusionsym}{\ensuremath{\oplus}}

\newcommand{\theagent}{\ensuremath{\agent{S}}}

\newcommand{\desideratumdiscount}[1]{\ensuremath{\mathbf{desD_{#1}}}}
\newcommand{\desideratumfusion}[1]{\ensuremath{\mathbf{desF_{#1}}}}
\newcommand{\requirementdiscount}[1]{\ensuremath{\mathbf{rd_{#1}}}}
\newcommand{\requirementfusion}[1]{\ensuremath{\mathbf{rf_{#1}}}}

\begin{document} 

\title{Subjective Logic Operators in Trust Assessment: an Empirical Study\thanks{The authors thank Lance Kaplan for the useful discussions. \qquad \qquad Research was sponsored by US Army Research laboratory and the UK Ministry of Defence and was accomplished under Agreement Number W911NF-06-3-0001. The views and conclusions contained in this document are those of the authors and should not be interpreted as representing the official policies, either expressed or implied, of the US Army Research Laboratory, the U.S. Government, the UK Ministry of Defense, or the UK Government. The US and UK Governments are authorized to reproduce and distribute reprints for Government purposes notwithstanding any copyright notation hereon.}}

\author{Federico Cerutti \and Alice Toniolo \and Nir Oren \and Timothy J. Norman}

\institute{F. Cerutti \and A. Toniolo \and N. Oren \and T. J. Norman \at
University of Aberdeen\\
School of Natural and Computing Science\\
King's College\\
AB24 3UE, Aberdeen, UK\\
\email{f.cerutti@abdn.ac.uk}\\
\email{alice.toniolo@abdn.ac.uk}\\
\email{n.oren@abdn.ac.uk}\\
\email{t.j.norman@abdn.ac.uk}} 

\date{Received: date / Accepted: date}
 
\maketitle

\begin{abstract} 
Computational trust mechanisms aim to produce trust ratings from
both direct and indirect information about agents'
behaviour. Subjective Logic (SL) has been widely adopted as the core
of such systems via its fusion and discount operators. In recent
research we revisited the semantics of these operators to explore an
alternative, geometric interpretation. In this paper we present a
principled desiderata for discounting and fusion operators in
SL. Building upon this we present operators that satisfy these
desirable properties, including a family of discount operators. We
then show, through a rigorous empirical study, that specific,
geometrically interpreted operators significantly
outperform standard SL operators in estimating ground truth. These
novel operators offer real advantages for computational models of
trust and reputation, in which they may be employed without
modifying other aspects of an existing system.

\end{abstract} 

\section{Introduction} 

Trust forms the backbone of human and artificial societies, improving
robustness of interactions by restricting the actions of untrusted
entities and mitigating the impact of untrusted information \citep{sensoy13reasoning}. Within
the multi-agent systems community \citep{Sabater2005}, the problem of
how to determine the degree of trustworthiness to assign to other
agents is foundational for the notion of agency and for its defining
relation of acting ``on behalf of''. Trustworthiness is utilised when
selecting partners for interactions; distrusted agents are less likely
to be engaged, reducing their influence over the system.

Trust mechanisms aim to compute a level of trust based on direct and
second-hand interactions between agents. The latter, commonly referred
to as \emph{reputational} information, is obtained from other agents
which have interacted with the subject of the assessment. Aspects of
such systems that have been examined include how to minimise the
damage caused by collusion between agents \citep{Haghpanah12prep}, the
nature of reputation information \citep{josang12trust}, and examining
trust in specific contexts and agent interaction configurations
\citep{burnett12subdelegation}.

In this paper we strengthen the analysis of an alternative to J{\o}sang's Subjective Logic (SL)  discounting and combination operators \citep{Josang2001}, which we have previously described in \citep{Cerutti2013a}.
In particular, we
enlarged the range of proposed discounting operators in order to
provide a more comprehensive experimental evaluation. Instead of
providing single operators we present a general approach, from which
an entire family of operators can be derived that are proved to be
compliant with specific desirable properties.  From our analysis we
can deduce some new interesting statistical properties of J{\o}sang's
operators, as well as of the proposed operators. This evaluation
methodology introduces two different metrics: the \emph{expected value
  distance} from the ground truth, and the \emph{geometric distance}
from the ground truth. According to the former, our family of
operators are shown to be almost equivalent to J{\o}sang's original operators,
and significantly more accurate in one case. Using the latter metric,
our operators compute reputation opinions closer to the ground truth
than J{\o}sang's, with the exception of one case. Further, one of our
proposed discounting operators outperforms the traditional SL operator
on both metrics.

In the next section we present a desiderata for discounting and fusion
operators, grounded on how trust models such as SL are
employed in practice.  After a brief overview of J{\o}sang's
SL in Section
\ref{sec:background-notions}, we formalise the desiderata in Section
\ref{sec:core-prop-requ} considering SL opinions, and show that they
are not satisfied by existing SL operators. We describe our proposed
operators in Section \ref{sec:core-properties}, and prove that they
comply with the desirable properties presented. Then, in Section
\ref{sec:experiment}, we describe our experiment designed to conduct
the comparative study among the operators, and, in Section
\ref{sec:results}, we evaluate our results to determine their
significance. We summarise the conclusions that can be drawn from this
study in Section \ref{sec:conclusions}. In order to improve the
readability of the paper, proofs of the described results can be found
in Appendix \ref{sec:proofs} (Appendix \ref{sec:geom-subj-logic}
discusses the mathematical foundations of the proposed operators).

\section{A Desiderata for Discounting and Fusion Operators}
\label{sec:desid-disc-fusi}

In this work we  focus on trust relations where an agent referred to as the \emph{truster}  --- $X$ --- ``depends'' on a \emph{trustee} --- $Y_i$ --- \citep{Castelfranchi2010}. As a concrete example, we examine the case where a trustee is responsible  for providing some information to a truster, as exemplified by the following scenario.

\begin{ex}[label=exa:scenario]
  \label{ex:scenario}
  Let $X$ be a military analyst who is collecting evidence in order to decide whether or not a specific area contains a certain type of weapon. In particular, he needs a datum \m{} from two sensors $Y_1$ and $Y_2$, each of which has a history of failure, thus affecting the trust $X$ places in them. Here, $X$ is the truster, and $Y_1$ and $Y_2$ the trustees.
\end{ex}

In such a scenario, the degree of trustworthiness of $Y_1$ and $Y_2$ is normally computed from historical data \citep{Josang2004}. Suppose that $X$ asks $Y_1$ and $Y_2$ about \m. 
Let us consider the case where $Y$ informs that it believes that \m{} holds (we write $Y$ to indicate either $Y_1$ or $Y_2$).

\begin{ex}[continues=exa:scenario]
  \textbf{Belief}. Suppose that $Y$ answers \linebreak that \m{} holds with absolute certainty. However, $Y$ has had a history of (random) failures, which means that $X$ does not completely trust $Y$'s reports. In this situation, it seems reasonable for $X$ to derive an \emph{opinion} about the likelihood of \m{} with an upper limit equivalent to its trust in $Y$.
\end{ex}

The above scenario gives us an intuition about a first desideratum concerning discounting opinions, viz.:
\begin{quotation}
  \desideratumdiscount{1}: when the trustworthiness degree of a trustee $Y$ is derived from historical data, if $Y$ informs $X$ that \m{} holds with absolute certainty, $X$ should believe \m{}  as much as it believes $Y$.
\end{quotation}
 
On the other hand, if $Y$ informs $X$ that it is uncertain about \m, this should be directly reflected in $X$'s opinion about \m.

\begin{ex}[continues=exa:scenario]
  \textbf{Uncertainty}. Consider the case \linebreak where $Y$ informs $X$ that it is unable to observe \m. Here, there is  complete uncertainty with regard to the degree of trustworthiness associated with $Y$'s reports about \m.
\end{ex}

We can thus derive an additional desideratum:

\begin{quotation}
  \desideratumdiscount{2}: if $Y$ informs $X$ that it is uncertain about \m, $X$ is also uncertain about \m.
\end{quotation}

In addition, we can identify an  intermediate case, where it is known that the current situation negatively affects the degree of trustworthiness, and where an estimate of this effect can be determined. This is illustrated in the following scenario.

\begin{ex}[continues=exa:scenario]
\textbf{Intermediate evidence}. Suppose \linebreak that $Y$ reports that its opinion about \m{} is not accurate, but there is some evidence in favour of \m{},  and some evidence against it (with some degree of uncertainty)\footnote{An example of this is GPS data, which is known to be inaccurate if you are using civilian equipment \citep{Bisdikian2012}.}. In this case $X$ knows that the data received is somewhat accurate, and can therefore derive a degree of trustworthiness in the information regarding \m{} received from $Y$.
\end{ex}

This illustrates another desideratum regarding the discounting of opinions, namely:

\begin{quotation}
  \desideratumdiscount{3}: when the trustworthiness degree of a trustee $Y$ is derived from historical data, if $Y$ informs $X$ that \m{} holds with some degree of certainty, $X$ should believe \m{} less its trust in $Y$. 
\end{quotation}

There are cases where the queried datum is not evidence about a physical phenomenon, but rather an opinion about another agent.

\begin{ex}[continues=exa:scenario]
\textbf{Reputation}. Suppose that 
sensor $Y_1$ provides information about readings obtained from $Y_{2}$,   an Uninterruptible Power Supply (UPS).. Further suppose that $X$ wants to query $Y_2$ about its battery status, but $X$ has, until now, only obtained its information directly from $Y_{1}$. Given a report from $Y_{2}$,  $X$ can ask $Y_1$ about $Y_2$'s degree of trustworthiness. According to our previous terminology, the message that $Y_1$ sends to $X$ is the subjective --- to $Y_1$ --- trustworthiness degree of $Y_2$.
\end{ex}

Finally, there are cases where $X$ has to integrate different sources of information.

\begin{ex}[continues=exa:scenario]
\textbf{Consensus}. Suppose that $X$ \linebreak queries both $Y_1$ and $Y_2$ about a physical phenomenon \m, and let us suppose that it receives different answers from the two sensors. In this situation, $X$ will search for a consensus between these opinions, and will be biased towards the answer obtained from the more historically accurate sensor.
\end{ex}

This illustrates a desideratum about fusing opinions:

\begin{quotation}
\desideratumfusion{1}: in the case where $X$ receives $n$ opinions about \m{} from $Y_1$, $Y_2$, \ldots $Y_n$, there must exist an operator informing which are $X$'s preferences among each $Y_i$ (i.e. related to their degree of trustworthiness), and $X$'s opinion about \m{} should be derived according to these preferences (i.e. giving more importance to the opinion received from the most preferred trustee).
\end{quotation}

\section{Background}
\label{sec:background-notions}

In the above scenario, we utilised the terms \emph{trust}, \emph{trustworthiness}, and \emph{reputation} which often have different meanings among different pieces of research. It is beyond the scope of this paper to investigate these meanings; the interested reader is referred to \cite{Castelfranchi2010} and \cite{Urbano2013} for an overview.
For the purpose of this paper we consider the notion of \emph{trustworthiness} as the property of an agent --- the \emph{trustee} --- that a \emph{trustor} is connected with, where this property represents the willingness of the trustee to share information accurately with respect to the ground truth (we make a distinction between inaccuracy that is intentional or otherwise). Moreover,  \emph{reputation} is a property which represents the subjective view of an arbitrary trustee's trustworthiness that we obtained  from an agent we can directly communicate with.

Following \citep{Josang2007}  we express both the \emph{degree of trustworthiness} and the \emph{degree of reputation} using SL. This formalism extends probability theory by expressing uncertainty about the probability values themselves, which makes it useful for representing trust degrees. We now proceed to provide a brief overview of SL mainly based on \citep{Josang2001}.

Like Dempster-Shafer evidence theory \citep{Dempster1968,Shafer1976}, SL  operates on a
\emph{frame of discernment}, denoted by $\Theta$. A frame of discernment
contains the set of possible system states, only one of which represents the
actual system state. These are referred to as atomic, or primitive, system states.
The powerset of
$\Theta$, denoted by $2^\Theta$, consists of all possible unions of primitive
states. A non-primitive state may contain other states within it. These are
referred to as substates of the state.

\begin{definition} \label{beliefMassAssignment}
Given a frame of discernment $\Theta$, we can associate a belief mass
assignment $m_\Theta(x)$ with each substate $x \in 2^\Theta$ such that
\begin{enumerate}
  \item $m_\Theta(x) \geq 0$
  \item $m_\Theta(\emptyset) = 0$
  \item $\displaystyle\sum_{x \in 2^\Theta} m_\Theta(x)=1$
\end{enumerate}
\end{definition}

For a substate $x$, $m_\Theta(x)$ is its \emph{belief mass}.

Belief mass is an unwieldy concept to work with. When we speak of belief in
a certain state, we refer not only to the belief mass in the state, but also
to the belief masses of the state's substates. Similarly, when we speak
about disbelief, that is, the total belief that a state is not true, we need
to take substates into account. Finally, SL also introduces the
concept of uncertainty, that is, the amount of belief that might be in a
superstate or a partially overlapping state. These concepts can be 
formalised as follows.

\begin{definition} Given a frame of
discernment $\Theta$ and a belief mass assignment $m_\Theta$ on $\Theta$, we
 define the belief function for a state $x$ as
$$b(x)=\sum_{y \subseteq x} m_\Theta(y) \textrm{ where } x,y \in 2^\Theta$$
The disbelief function as
$$d(x) = \sum_{y \cap x=\emptyset} m_\Theta(y) \textrm{ where } x,y \in
2^\Theta$$
And the uncertainty function as
$$u(x) = \sum_{\tiny{\begin{array}{l}y \cap x \neq \emptyset \\ y \not \subseteq x
\end{array}}} m_\Theta(y) \textrm{ where } x,y \in 2^\Theta$$
\end{definition}

These functions have two important properties. First, they
all range between zero and one. Second, they always sum to one, meaning that
it is possible to deduce the value of one function given the other two. 

Boolean logic operators have SL equivalents. It makes sense to
use these equivalent operators in frames of discernment containing a state and
(some form of) the state's negation. A \emph{focused frame of discernment} is
a binary frame of discernment containing a state and its complement.

\begin{definition} Given $x \in
2^\Theta$, the frame of discernment denoted by $\tilde{\Theta}^x$, which contains
two atomic states, $x$ and $\neg x$, where  $\neg x$ is the complement of $x$
in $\Theta$, is the focused frame of discernment with focus on $x$. 

Let $\tilde{\Theta}^x$ be the focused frame of discernment with
focus on x of $\Theta$. Given a belief mass assignment $m_\Theta$ and the
belief, disbelief and uncertainty functions for $x$ ($b(x)$, $d(x)$ and $u(x)$
respectively), the focused belief mass assignment, $m_{\tilde{\Theta}^x}$ on
$\tilde{\Theta}^x$ is defined as
\begin{eqnarray*}
&& m_{\tilde{\Theta}^x}(x) = b(x) \\
&& m_{\tilde{\Theta}^x}(\neg x) = d(x) \\
&& m_{\tilde{\Theta}^x}(\tilde{\Theta}^x) = u(x) 
\end{eqnarray*}
The focused relative atomicity of $x$ (which approximates the role of a prior probability distribution within probability theory, weighting the likelihood of some outcomes over others) is defined as
$$a_{\tilde{\Theta}^x}(x/\Theta)=[E(x)-b(x)]/u(x)$$
\noindent
where $E(x)$ represents the \emph{expected value} of $x$.

For convenience, the focused relative atomicity of $x$ is often abbreviated $a_{\tilde{\Theta}^x}(x)$ or $a(x)$.
\end{definition}

An opinion consists of the belief, disbelief, uncertainty and relative
atomicity as computed over a focused frame of discernment.

\begin{definition}
Given a focused frame of discernment $\Theta$ containing $x$ and its complement
$\neg x$, and assuming a belief mass assignment $m_\Theta$ with belief,
disbelief, uncertainty and relative atomicity functions on $x$ in $\Theta$ of
$b(x)$,$d(x)$,$u(x)$ and $a(x)$, we define an \emph{opinion} over $x$, written $\omega_x$ as
$$\omega_x \equiv \langle b(x),d(x),u(x),a(x) \rangle$$
\end{definition}

For compactness, J{\o}sang also denotes the various functions as
$b_x$,$d_x$,$u_x$ and $a_x$ in place, and we will follow this notation. Furthermore, given a fixed $a_x$, an opinion $\omega$ can be denoted as a $\langle b_{x},d_{x},u_{x} \rangle$ triple.

Given opinions about two propositions from different frames of discernment, it
is possible to combine them in various ways using SL's various operators, as detailed in \citep{Josang2001,Josang2004,Josang2005,Josang2006,McAnally2004}. 
In this work we concentrate on J{\o}sang's \emph{discount} and \emph{fusion} operators, which we review next.

\begin{definition}[Def. 5, \cite{Josang2006}]
  \label{def:josang-discount}
  Let $A, B$ be two agents where $A$'s opinion about $B$ is expressed as $\omega_B^A = \tuple{b_B^A, d_B^A, u_B^A, a_B^A}$ and let $x$ be a proposition where $B$'s opinion about $x$ is acquired by $A$ as the opinion $\omega_x^B =\linebreak \tuple{b_x^B, d_x^B, u_x^B, a_x^B}$. Let $\omega_x^{A:B} = \tuple{b_x^{A:B}, d_x^{A:B}, u_x^{A:B}, a_x^{A:B}}$ be the opinion such that:

\[
\begin{cases} 
  b_x^{A:B} = b_B^A~b_x^B\\
  d_x^{A:B} = b_B^A~d_x^B\\
  u_x^{A:B} = d_B^A + u_B^A + b_B^A~u_x^B\\
  a_x^{A:B} = a_x^B\\
\end{cases}
\]

\noindent
then $\omega_x^{A:B}$ is called the \emph{uncertainty favouring discounted opinion of $A$}. By using the symbol $\josangdiscountsym$ to designate this operation, we get $\omega_x^{A:B} = \omega_B^A \josangdiscountsym \omega_x^B$.
\end{definition}

\begin{definition}[Thm. 1, \cite{Josang2006}]
  \label{def:josang-fusion}
  Let $\omega_x^A = \linebreak \tuple{b_x^A, d_x^A, u_x^A, a_x^A}$ and $\omega_x^B = \tuple{b_x^B, d_x^B, u_x^B, a_x^B}$ be trust in $x$ from $A$ and $B$ respectively. The opinion $\omega_x^{A\diamond B}= \linebreak \tuple{b_x^{A\diamond B}, d_x^{A\diamond B}, u_x^{A\diamond B}, a_x^{A\diamond B}}$ is then called the consensus between $\omega_x^A$ and $\omega_x^B$, denoting the trust that an imaginary agent $[A, B]$ would have in $x$, as if that agent represented both $A$ and $B$. In case of Bayesian (totally certain) opinions, their relative weight can be defined as $\gamma^{A/B} = \lim{(u_x^B/u_x^A)}$.

  Case I: $u_x^A + u_x^B - u_x^A~u_x^B \neq 0$

  $\begin{cases} 
    b_x^{A\diamond B} = \frac{b_x^A~u_x^B + b_x^B~u_x^A}{u_x^A + u_x^B - u_x^A~u_x^B}\\
    d_x^{A\diamond B} = \frac{d_x^A~u_x^B + d_x^B~u_x^A}{u_x^A + u_x^B - u_x^A~u_x^B}\\
    u_x^{A\diamond B} = \frac{u_x^A~u_x^B}{u_x^A + u_x^B - u_x^A~u_x^B}\\
    a_x^{A\diamond B} = \frac{a_x^A~u_x^B + a_x^B~u_x^A - (a_x^A + a_x^B)~u_x^A~u_x^B}{u_x^A + u_x^B - 2~u_x^A~u_x^B}\\
  \end{cases}$

  Case II: $u_x^A + u_x^B - u_x^A~u_x^B = 0$

  $\begin{cases} 
    b_x^{A\diamond B} = \frac{(\gamma^{A/B}~b_x^A + b_x^B)}{(\gamma^{A/B} + 1)}\\
    d_x^{A\diamond B} = \frac{(\gamma^{A/B}~d_x^A + d_x^B)}{(\gamma^{A/B} + 1)}\\
    u_x^{A\diamond B} = 0\\
    a_x^{A\diamond B} = \frac{(\gamma^{A/B}~a_x^A + a_x^B)}{(\gamma^{A/B} + 1)}\\
  \end{cases}$

By using the symbol `$\josangfusionsym$' to designate this operator, we can write $\omega_x^{A\diamond B} = \omega_x^A \josangfusionsym \omega_x^B$.
\end{definition}

\section{Core Properties and Requirements}
\label{sec:core-prop-requ}

In our scenario, agent $X$ has to determine the trustworthiness degree associated with \m{} received from $Y_i$. $X$ will consider three elements in reaching a decision:

\begin{enumerate}
\item \textbf{trustworthiness:} $X$ has an opinion $\T{i}$ concerning the degree of trustworthiness of $Y_i$;

\item \textbf{certainty:} $Y_i$ communicates that \m{} holds with a degree of certainty $\C{i}$;

\item \textbf{combination:} $X$ has to combine $\T{1}, \ldots, \T{n} $ with (resp.) $\C{1}, \ldots, \C{n}$ in order to achieve an ultimate opinion $\W{} = \fusionop{}(\W{1}, \ldots, \W{n})$ on \m, where $\forall i \W{i} = \optrustconf{\T{i}}{\C{i}}{}$, \ie{} $\W{i}$ is the result of a combination of opinion $\T{i}$ with opinion $\C{i}$, and each opinion $\W{i} = \optrustconf{\T{i}}{\C{i}}{}$ has associated a weight $K_i = f(\T{i})$ for some function $f(\cdot)$. 
\end{enumerate}

In particular, the three desiderata for discounting opinions (\desideratumdiscount{1}, \desideratumdiscount{2}, \desideratumdiscount{3}), give rise to the following  \emph{requirements for  discounting}: 
\begin{itemize}
\item[\requirementdiscount{1}] if $\C{i}$ is pure belief, then $\W{i} = \T{i}$;
\item[\requirementdiscount{2}] if $\C{i}$ is completely uncertain, then $\W{i} = \C{i}$ (\ie{} $\langle 0,0,1\rangle$);
\item[\requirementdiscount{3}] the degree of belief of $\W{i}$ is always less than or equal to the degree of belief of $\T{i}$;
\end{itemize}

\noindent
and the desideratum about fusing opinion (\desideratumfusion{1}) gives rise to the following \emph{requirements for fusion}:
\begin{itemize}
\item[\requirementfusion{1}] if $\forall i, j$ $K_i = K_j$, then $\W{} = \fusionop{}(\W{1}, \ldots, \W{n})$ is the opinion resulting from the average of $\W{1}, \ldots, \W{n}$ ;
\item[\requirementfusion{2}] if $\exists i$ \suchthat{} $K_i = 0$, then $\fusionop{}(\W{1}, \ldots, \W{n}) = \linebreak \fusionop{}(\W{1}, \ldots, \W{i-1}, \W{i+1}, \ldots, \W{n})$.
\end{itemize}

While there is a direct correspondence between the de\-si\-de\-ra\-tum for discounting opinion (\desideratumdiscount{1}, \desideratumdiscount{2}, \desideratumdiscount{3}) and the requirements for discounting (\requirementdiscount{1}, \requirementdiscount{2}, \requirementdiscount{3}), the desideratum about fusing opinion \desideratumfusion{1} gives rise to two (loose) requirements, \requirementfusion{1} and \requirementfusion{2}.
Note  that the combination requirements describe the same ``prudent'' behaviour as was presented in Example \ref{ex:scenario}, in particular in the ``belief'' scenario. Indeed even if $X$ is highly confident in a specific context, this confidence cannot increase the trust degree over the base trustworthiness degree.

Following \citep{McAnally2004,Josang2006}, we utilise SL to instantiate trustworthiness and confidence, and seek to compute their combination through SL operations\footnote{Hereafter each opinion will have a fixed relative atomicity of $\frac{1}{2}$. This assumption will be relaxed in future works.}. In doing so, we must therefore consider the following inputs and requirements:

\begin{enumerate}
\item $\opinionset = \set{\opinion{} \in \realset^3 | 0 \leq \belief{} \leq 1 \mbox{ and } 0 \leq \disbelief{} \leq 1 \mbox{ and } 0 \leq \uncertainty{} \leq 1 \mbox{ and } \belief{} + \disbelief{} + \uncertainty{} = 1}$;
\item $\T{i} = \opinion{\T{i}}$ derived by statistical observations (\eg{} \citep{McAnally2004});
\item $\C{i} = \opinion{\C{i}}$ is the opinion received from $Y_i$ about \m;
\item $\optrustconf{}{}{}: \opinionset \times \opinionset \mapsto \opinionset$;
\item $\fusionop{}: \opinionset^n \mapsto \opinionset$;
\item $\W{i} = \left\{ 
    \begin{array}{l l}
      \T{i} & \mbox{if}~\C{i} = \tuple{1, 0, 0}\\
      \C{i} & \mbox{if}~\C{i} = \tuple{0, 0, 1}\\
    \end{array}
  \right.$

  further requiring that  $\belief{\W{i}} \leq \belief{\T{i}}$;
\item given $\W{} = \fusionop{}(\W{1}, \ldots, \W{n})$, if $\forall i,j \in \set{1, \ldots, n} K_i = K_j$, then 
  \begin{itemize}
  \item $\belief{\W{}} = \displaystyle{\frac{1}{n} \sum_{i = 1}^n \belief{\W{i}}}$
  \item $\disbelief{\W{}} = \displaystyle{\frac{1}{n} \sum_{i = 1}^n \disbelief{\W{i}}}$
  \item $\uncertainty{\W{}} = \displaystyle{\frac{1}{n} \sum_{i = 1}^n \uncertainty{\W{i}}}$
  \end{itemize}
\item if $\exists i$ \suchthat{} $K_i = 0$, then $\fusionop{}(\W{1}, \ldots, \W{n}) = \linebreak \fusionop{}(\W{1}, \ldots, \W{i-1}, \W{i+1}, \ldots, \W{n})$.
\end{enumerate}

Although any function $f(\cdot)$ can be used for deriving $K_i$, hereafter we will consider $\T{i}$'s expected value, i.e. $K_i = \belief{T_i} + \frac{\uncertainty{T_i}}{2}$.

Since 1--5 above are inputs, we concentrate on the constraints expressed by 6, 7 and 8, which require us to consider the problem of how to combine the degree of trustworthiness with the degree of certainty. Existing work, such as \citep{McAnally2004,Castelfranchi2010,Urbano2013}, concentrate on computing $\T{}$. 

We begin by noting --- as illustrated in Table \ref{tab:comparison-josang-operators} --- that no set of  operators provided by SL  \citep{Josang2001,Josang2004,Josang2005,Josang2006,McAnally2004,Josang2008} satisfies the combination requirements described previously.

\begin{table}[h]
  \centering
  \begin{tabular}{l | c | c | c | c | c }
                   & \multicolumn{5}{c}{Requirement satisfied?}\\
                   & \multicolumn{3}{c|}{Discount req.} & \multicolumn{2}{c }{Fusion req.}\\
    Operator       & \requirementdiscount{1} & \requirementdiscount{2} & \requirementdiscount{3} & \requirementfusion{1} & \requirementfusion{2} \\
    \hline
    Addition ($+$) & \cellcolor{lightgray}No & \cellcolor{lightgray}No & \cellcolor{lightgray}No & \cellcolor{lightgray}No & \cellcolor{lightgray}No\\
    Subtraction ($-$) & \cellcolor{lightgray}No & \cellcolor{lightgray}No & Yes & \cellcolor{lightgray}No & \cellcolor{lightgray}No\\
    Multiplication ($\cdot$) & \cellcolor{lightgray}No & \cellcolor{lightgray}No & \cellcolor{lightgray}No & \cellcolor{lightgray}No & \cellcolor{lightgray}No\\
    Division ($/$) & \cellcolor{lightgray}No & \cellcolor{lightgray}No & \cellcolor{lightgray}No & \cellcolor{lightgray}No & \cellcolor{lightgray}No\\
    Comultiplication ($\sqcup$) & \cellcolor{lightgray}No & \cellcolor{lightgray}No & \cellcolor{lightgray}No & \cellcolor{lightgray}No & \cellcolor{lightgray}No\\
    Codivision ($\sqcap$) & \cellcolor{lightgray}No & \cellcolor{lightgray}No & \cellcolor{lightgray}No & \cellcolor{lightgray}No & \cellcolor{lightgray}No \\
    Discounting ($\otimes$) & \cellcolor{lightgray}No & Yes & Yes & \cellcolor{lightgray}No & \cellcolor{lightgray}No\\
    Cumulative fusion ($\oplus$) & \cellcolor{lightgray}No & \cellcolor{lightgray}No & \cellcolor{lightgray}No & \cellcolor{lightgray}No & \cellcolor{lightgray}No\\
    Averaging fusion ($\underline{\oplus}$) & \cellcolor{lightgray}No & \cellcolor{lightgray}No & \cellcolor{lightgray}No & Yes & \cellcolor{lightgray}No \\
    Cumulative unfusion ($\ominus$) & \cellcolor{lightgray}No & \cellcolor{lightgray}No & \cellcolor{lightgray}No & \cellcolor{lightgray}No & \cellcolor{lightgray}No\\
    Averaging unfusion ($\underline{\ominus}$) & \cellcolor{lightgray}No & \cellcolor{lightgray}No & \cellcolor{lightgray}No & \cellcolor{lightgray}No & \cellcolor{lightgray}No \\
  \end{tabular}
  \caption{J{\o}sang operators and the satisfaction of the five combination requirements}
  \label{tab:comparison-josang-operators}
\end{table}

In the next section we describe our proposals for the discount --- $\optrustconf{}{}{}$ --- and consensus --- $\fusionop{}$ --- operators in order to satisfy the above five requirements. 


\section{The  Operators}
\label{sec:core-properties}


\subsection{A Na\"ive Discount Operator}
\label{sec:naive-disc-oper}

As suggested by an anonymous referee of a preliminary version of \citep{Cerutti2013}, a very na\"ive operator satisfying the requirements \requirementdiscount{1}, \requirementdiscount{2}, and \requirementdiscount{3} is the following.

\begin{definition}
  \label{def:naive}
  Given the two opinions $\T{} = \opinion{\T{}}$ and $\C{} = \opinion{\C{}}$, the na\"ive-discount of $\C{}$ by $\T{}$ is $\W{} = \optrustconf{\T{}}{\C{}}{n}$, where:
  \begin{itemize}
  \item $\belief{\W{}} = \belief{\C{}} \cdot \belief{\T{}}$;
  \item $\disbelief{\W{}} = \belief{\C{}} \cdot \disbelief{\T{}} + \disbelief{\C{}}$;
  \item $\uncertainty{\W{}} = \belief{\C{}} \cdot \uncertainty{\T{}} + \uncertainty{\C{}}$.
  \end{itemize}
\end{definition}

The following proposition shows that the na\"ive discount operator fulfils the first three requirements.

\begin{restatable}{propn}{propnaive}
\label{prop:naive}
Given the two opinions $\T{} = \opinion{\T{}}$ and $\C{} = \opinion{\C{}}$, and $\W{} = \optrustconf{\T{}}{\C{}}{n}$, the na\"ive-discount of $\C{}$ by $\T{}$, then:
\begin{itemize}
\item[$i$.] $\W{} \in \opinionset$;
\item[$ii$.] if $\C{} = \tuple{1, 0, 0} $, then $\W{} = \T{}$, \ie{} requirement \requirementdiscount{1};
\item[$iii$.] if $\C{} = \tuple{0, 0, 1} $, then $\W{} = \C{}$, \ie{} requirement \requirementdiscount{2};
\item[$iv$.] $\belief{\W{}} \leq \belief{\T{}}$, \ie{} requirement \requirementdiscount{3}.
\end{itemize}
\end{restatable}

However, one of the limits of this na\"ive operator is that if $\C{} = \tuple{0, 1, 0}$, the discounted opinion is pure disbelief. This means that, regardless of the trustworthiness opinion an agent has on a source of information, if this source of information informs the agent that it is certain in its disbelief of a message, then the result of the discounting action is complete disbelief.

Although this seems to be reasonable in some contexts, e.g.~in the merging of confidence and trustworthiness opinions (as discussed in \cite{Cerutti2013}), it can be questionable in the context of discounting opinions \citep{Kaplan13personal} --- 
intuition suggests discounting an opinion should (generally) raise the degree of uncertainty, while the na\"ive operator reduces it.

In the following section we introduce a family of discount operators, which can provide varying degree of uncertainty when discounting opinions.

\subsection{A Family of Graphical Discount Operators}
\label{sec:graph-disc-oper}

As discussed in \citep{Josang2001}, a Subjective Logic opinion admits a geometrical representation inside a triangle, and, as shown in \citep{Josang2005}, operators can be defined in order to satisfy graphical properties\footnote{\citep{Josang2005} shows how the deduction operator affects the space of possible derived opinions.}.

As detailed in Appendix \ref{sec:geom-subj-logic}, using the constraint that an opinion's belief, disbelief and uncertainty must sum to 1, we can flatten the 3-dimension space \opinionset{}  into a 2-dimension space (Cartesian space). 

\begin{figure}[h]
  \centering
  \includegraphics[scale=0.7]{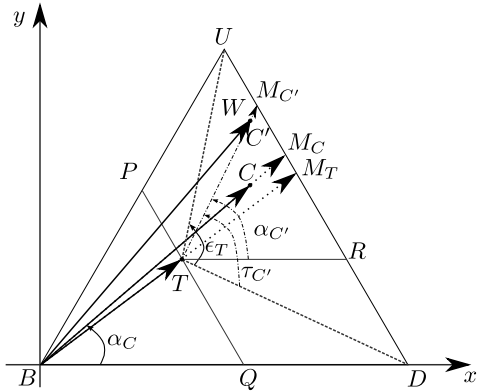}
  \caption{Projection of the certainty opinion and its combination with a trustworthiness opinion}
  \label{fig:vector-confidence-combination}
\end{figure}

Figure \ref{fig:vector-confidence-combination} depicts the intuition behind the family of discount operators we introduce in this paper. In the figure, the point \T{} represents the subjective logic opinion regarding the trustworthiness degree of the source of information. Given this point, the four-sided figure $PQDU$ represents the \emph{admissible space of opinions}, where $P \triangleq \tuple{\belief{\T{}}, 0, 1 - \belief{\T{}}}, Q \triangleq \tuple{\belief{\T{}}, 1 - \belief{\T{}}, 0}, D \triangleq \tuple{0, 1, 0}, U \triangleq \tuple{0, 0, 1}$. In other words, the opinions that are inside the admissible space clearly satisfy requirement \requirementdiscount{3}, as their degree of belief cannot be greater than $\belief{\T{}}$.

\begin{definition}
  \label{defn:admissiblespace}
  Given an opinion $\T{} = \opinion{\T{}}$, the \emph{admissible space of opinions given \T{}} is $\admissibleopinionset{\T{}} = \set{X \in \opinionset{} | \belief{X} \leq \belief{\T{}}}$.
\end{definition}

We can easily show that the line between $P$ and $Q$ as shown in Fig. \ref{fig:vector-confidence-combination} delimits the admissible space of opinions for $\T{}$.

\begin{restatable}{propn}{propadmissiblespace}
\label{propn:admissible-space}
Given an opinion $\T{} = \opinion{\T{}}$, and its Cartesian representation, the four-sided figure $PQDU$ represents $\admissibleopinionset{\T{}}$, where:
\begin{itemize}
\item $P \triangleq \tuple{\belief{\T{}}, 0, 1-\belief{\T{}}}$; and
\item $Q \triangleq \tuple{\belief{\T{}}, 1-\belief{\T{}}, 0}$.
\end{itemize}
\end{restatable}

The idea behind the discount operator is as follows. An opinion $\C{}$ should be projected into the admissible space of opinions given \T{}. According to Fig. \ref{fig:vector-confidence-combination}, discounting the opinion $\C{}$ with the opinion $\T{}$ means that we project $\C{}$ into $\admissibleopinionset{\T{}}$ thus achieving a new opinion $\C{}'$ which is the result of the discounting operator.
We consider only a linear projection in this work, but more complex functions can be easily envisaged. In other words, the idea is that $|\vv{B\C{}}| \propto |\vv{\T{}\C{}'}|$, as well as the angle $\alpha_{\C{}} \propto \alpha_{\C{}'}$.


The following definition describes the family of discount operators obtained following the above intuition. Each member of the family is identified by a specific value of $\alpha_{\C{}'}$.

\begin{definition}
\label{defn:combination}
  Given the two opinions $T = \opinion{T}$ and $C = \opinion{C}$, the graphical-discount of $C$ by $T$ is $W = \optrustconffam{T}{C}{}$, where:
  \begin{itemize}
  \item $\uncertainty{W} = \uncertainty{T} + \sin(\alpha_{C'}) |\vv{TC'}|$
  \item $\disbelief{W} = \disbelief{T} + (\uncertainty{T} - \uncertainty{W}) \cos(\frac{\pi}{3}) + \cos(\alpha_{C'}) \sin(\frac{\pi}{3}) |\vv{TC'}|$
  \end{itemize}

  \noindent
  where\footnote{With reference to Fig. \ref{fig:vector-confidence-combination},  $\alpha_C \triangleq \ang{CBD}$, $\beta_T \triangleq \ang{TDB}$, $\gamma_T \triangleq \ang{TDU}$, $\delta_T \triangleq \ang{TUD}$, $\epsilon_T \triangleq \ang{DTU}$.}:
  \begin{itemize}
  \item $\displaystyle{\frac{\alpha_C~\epsilon_T}{\frac{\pi}{3}} - \beta_T \leq \alpha_{\C{}'} \leq \epsilon_{\T{}} - \beta_{\T{}}}$
  \item $\displaystyle{\alpha_C} = \left\{
        \begin{array}{l l}
          0 & \mbox{if } \belief{C} = 1\\
          \displaystyle{\arctan\left(\frac{\uncertainty{C}~\sin(\frac{\pi}{3})}{\disbelief{C} + \uncertainty{C}~\cos(\frac{\pi}{3})}\right)} & \mbox{otherwise}\\
        \end{array} \right.$
  \item $\displaystyle{\beta_T } = \left\{
      \begin{array}{l l}
        \displaystyle{\frac{\pi}{3}} & \mbox{if } \disbelief{T} = 1\\
        \displaystyle{\arctan\left(\frac{\uncertainty{T}~\sin(\frac{\pi}{3})}{1-(\disbelief{T}+\uncertainty{T}~\cos(\frac{\pi}{3}))}\right)} & \mbox{otherwise}\\
      \end{array}\right.$
  \item $\displaystyle{\gamma_T = \frac{\pi}{3}} - \beta_T$
  \item $\displaystyle{\delta_T = }\left\{
        \begin{array}{l l}
          0 & \mbox{if } \uncertainty{T} = 1\\
          \displaystyle{\arcsin\left(\frac{\belief{T}}{|\vv{TU}|}\right)} & \mbox{otherwise}\\
        \end{array}\right.$
  \item $\epsilon_T = \pi - \gamma_T - \delta_T$

  \item $\displaystyle{|\vv{TC'}| = \frac{|\vv{BC}|}{|\vv{BM_{C}}|} |\vv{TM_{C'}}| =}$ \\ 
$=\displaystyle{r_C~|\vv{TM_{C'}}|}$
  \end{itemize}

  \noindent
  with $r_C = \frac{|\vv{BC}|}{|\vv{BM_C}|}$, and 

  \noindent
  $|\vv{TM_{C'}}| = \left\{
    \begin{array}{l l}
      2~\belief{T} & \mbox{if } \displaystyle{\alpha_{C'} = \frac{\pi}{2}}\\
      \displaystyle{\frac{2}{\sqrt{3}}~\uncertainty{T}} & \mbox{if } \displaystyle{\alpha_{C'} = - \frac{\pi}{3}}\\
      \displaystyle{\frac{2}{\sqrt{3}}~(1 - \uncertainty{T})} & \mbox{if } \displaystyle{\alpha_{C'} = \frac{2}{3} \pi}\\
      \displaystyle{\frac{2 \sqrt{\tan^2(\alpha_{C'}) + 1}}{| \tan(\alpha_{C'}) + \sqrt{3} | } ~\belief{T}} & \mbox{otherwise}\\
    \end{array}\right.$
\end{definition}

We can now show that this family of operators satisfy the requirements \requirementdiscount{1}, \requirementdiscount{2}, and \requirementdiscount{3}.

\begin{restatable}{thm}{thmfamilydiscount}
\label{thm:familydiscount}
  Given $\T{} = \opinion{\T{}}$, $C = \opinion{\C{}}$, and $\W{} = \optrustconffam{\T{}}{\C{}}{}$, then:
  \begin{itemize}
  \item[i.] $\W{} = \opinion{\W{}}$ is an opinion;
  \item[ii.] if $\C{} = \tuple{1, 0, 0}$, then $\W{} = \T{}$;
  \item[iii.] if $\C{} = \tuple{0, 0, 1}$, then $\W{} = \C{}$;
  \item[iv.] $\belief{\W{}} \leq \belief{\T{}}$.
  \end{itemize}
\end{restatable}

In order to  empirically examine the difference among members of this family of operators, let us define the following three graphical discount operators. The first,  \optrustconf{}{}{1}, considers the widest range of $\alpha_{\C{}'}$ possible, projecting $\C{}$ in the triangle $\tri{DTU}$. This operator is a ``geometrical counterpart'' to the na\"ive operator --- projecting $\C{}$ in $\tri{DTU}$ could result in an opinion with less uncertainty than $\T{}$. The second operator, \optrustconf{}{}{2}, raises the uncertainty of the discount opinion projecting $\C{}$ in the triangle $\tri{RTU}$. The third operator, \optrustconf{}{}{3}, considers the triangle $\tri{STU}$, where $S$ is the intersection of the line $DU$ with the bisector of the angle $\epsilon_T$ (this point is not shown in Fig. \ref{fig:vector-confidence-combination}).

\begin{definition}
  \label{def:graphicaldiscountoperators}
  Given the two opinions $T = \opinion{T}$ and $C = \opinion{C}$:
  \begin{itemize}
  \item $\optrustconf{}{}{1}$ is $\optrustconffam{}{}{}$ with $\displaystyle{\alpha_{\C{}'} = \frac{\alpha_{\C{}} \epsilon_{\T{}}}{\frac{\pi}{3}} - \beta_{\T{}}}$;
  \item $\optrustconf{}{}{2}$ is $\optrustconffam{}{}{}$ with $\displaystyle{\alpha_{\C{}'} = \frac{\alpha_{\C{}} (\epsilon_{\T{}} - \beta_{\T{}})}{\frac{\pi}{3}}}$;
  \item $\optrustconf{}{}{3}$ is $\optrustconffam{}{}{}$ with $\displaystyle{\alpha_{\C{}'} = \frac{\alpha_{\C{}} \frac{\epsilon_{\T{}}}{2}}{\frac{\pi}{3}} + \frac{\epsilon_{\T{}}}{2}- \beta_{\T{}}}$.
  \end{itemize}
\end{definition}

\subsection{The Graphical Fusion Operator}
\label{sec:graph-fusi-oper}

In \citep{Cerutti2013a} we introduced a fusion operator which satisfies requirements \requirementfusion{1} and \requirementfusion{2}.
Let us suppose we have $n$ opinions $W_1, W_2, \ldots, W_n$ derived using an operator $\circ$ \suchthat{} $\forall i \in \set{1, \ldots, n}, W_i = T_i \circ C_i$.  Intuitively, the fused opinion $\fusionop{}(W_1, W_2, \ldots, W_n)$ we want to obtain is the  ``balanced'' centroid of the polygon determined by the $n$ opinions.

\begin{definition}
  \label{def:graphicalfusion}
  Given the opinions $T_1, T_2, \ldots, T_n, \linebreak C_1, C_2, \ldots, C_n, W_1, W_2, \ldots, W_n$ \suchthat{} $\forall i \in \set{1\ldots n}, W_i = T_i \circ C_i$, the opinion resulting from the fusion of opinions $W_1, W_2, \ldots, W_n$ is \linebreak $\opinion{\fusionop{1}(W_1, \ldots, W_n)}$ where:

\begin{itemize}
\item $\displaystyle{\belief{\fusionop{1}(W_1, \ldots, W_n)} = \frac{1}{\sum_{i=1}^n K_i} \left(\sum_{i=1}^n K_i~ \belief{W_i}\right)}$
\item $\displaystyle{\disbelief{\fusionop{1}(W_1, \ldots, W_n)} = \frac{1}{\sum_{i=1}^n K_i} \left(\sum_{i=1}^n K_i~ \disbelief{W_i}\right)}$
\item $\displaystyle{\uncertainty{\fusionop{1}(W_1, \ldots, W_n)} = \frac{1}{\sum_{i=1}^n K_i} \left(\sum_{i=1}^n K_i~ \uncertainty{W_i}\right)}$
\end{itemize}
\end{definition}

This definition of $\fusionop{1}$ satisfies the requirements \requirementfusion{1} and \requirementfusion{2}. Moreover, the following proposition shows that \linebreak $\fusionop{1}(\W{1},\ldots,\W{n})$ is an opinion, and its Cartesian representation is the balanced centroid of the polygon identified by the points $\W{1}, \ldots, \W{n}$.

\begin{restatable}{propn}{fusionprop}
  \label{propn:fusion}
  Given the opinions $T_1, T_2, \ldots, T_n,$ $C_1, \linebreak C_2, \ldots, C_n,$ $W_1, W_2, \ldots, W_n$ \suchthat{} $\forall i \in \set{1\ldots n}, W_i = T_i \circ C_i$, and  $\opinion{\fusionop{1}(W_1, \ldots, W_n)}$ the opinion resulting from the fusion of opinions $W_1, W_2, \ldots, \linebreak W_n$, then:

  \begin{itemize}
  \item[i.] $\opinion{\fusionop{1}(W_1, \ldots, W_n)}$ is an opinion
  \item[ii.] $\displaystyle{x_{\fusionop{1}(W_1, \ldots, W_n)} = \frac{1}{\sum_{i=1}^n K_i} \left(\sum_{i=1}^n K_i~ x_{W_i}\right)}$
  \item[iii.] $\displaystyle{y_{\fusionop{1}(W_1, \ldots, W_n)} = \frac{1}{\sum_{i=1}^n K_i} \left(\sum_{i=1}^n K_i~ y_{W_i}\right)}$
  \end{itemize}
\end{restatable}

\section{Experimental Evaluation}
\label{sec:experiment}

In Section \ref{sec:desid-disc-fusi} we discussed the desiderata for the discounting and fusion operators. In Section \ref{sec:core-prop-requ}, we used these to obtain requirements \requirementdiscount{1}, \requirementdiscount{2}, \requirementdiscount{3}, \requirementfusion{1}, \requirementfusion{2}, which our proposed operator has been shown to fulfil.

In this section we focus on the design of an experiment aimed at evaluating, from an empirical point of view, the significance and usefulness of the desiderata and the requirements we identified for discounting and fusion operators. To this end, we consider an experiment where an explorer agent has to explore a network in order to determine the trustworthiness degree of other members: some of them are directly connected to the explorer, while for others it has to discount reputational information it obtains in the form of opinions. The evaluation is based upon the distance between the derived opinion and the ground truth, i.e.~the probability that each agent operates in a trustworthy manner. This likelihood is instantiated with the system, and is not known by anyone else in the network. This is just one of the possible scenario where the discounting and fusion operators can be applied: other empirical evaluations are left as avenues of future research.

The experimental setup follows the one described in \citep{Cerutti2013a}: changes are highlighted in the text.
In this experiment each agent can communicate with all the other agents in the network.
In order to randomly generate these networks, we consider a variable $\linkprob \in [0,1]$ representing the probability that an agent is connected\footnote{The term ``connected'' here can have different  names in different contexts, like ``friend'' in Facebook, or ``follower'' in Twitter.} to another agent (we exclude self-connections). Note that we do not constrain connections to be bidirectional\footnote{Although this may seem counter-intuitive, it partially captures real-world social media. For instance Twitter messages are public, therefore we don't know who will read our messages. The same applies with slightly modifications to Google+, and, of course, to blogging activities in general.}. In this experiment we set $\linkprob$ to lie between 5 and 25 with  increments  of 5. For each of the value of $\linkprob$, we execute the following procedure.

\subsection{Trust System Construction}

We build a set of 50 agents $\setagents = \set{\agent{1}, \ldots, \agent{50}}$: each agent $\agent{x}$ is characterised by a knowledge base $\setbeliefs{\agent{x}}$ and by the probability of responding truthfully to another agent's query, namely $\truthprob{\agent{x}} \in [0\ldots 1]$. For each agent $\agent{x}$, we randomly choose $\truthprob{\agent{x}} $.

We also require that $(\sharedbelief = \truesymbol) \in \setbeliefs{\agent{x}}$. In other words, all the agents share the same information $(\sharedbelief = \truesymbol)$ to be read ``\agent{x} knows that $\sharedbelief$ is $\truesymbol$''.

For each agent $\agent{x}$, we determine if it can communicate with $\agent{y} \neq \agent{x}$ according to $\linkprob$: if $\agent{y}$ is connected to $\agent{x}$, then we say that $\agent{y}$ is a connection of $\agent{x}$ ($\agent{y} \in \neighbours{\agent{x}}$).

Following the construction of the system, the experiment proceeds through two distinct phases, namely bootstrapping and exploration.

\subsection{Phase I: Bootstrapping} 

The bootstrapping phase is similar to that described in  \cite{Ismail2002}, where a $\beta$ distribution is used for analysing repetitive experiments and deriving a SL opinion. In this experiment, each agent $\agent{x}$ asks each of its connections $\agent{y} \in \neighbours{\agent{x}}$ about  $\sharedbelief$ a number of times equals to $\lboot$. Each time, $\agent{y}$ provides a possible false answer, based on the probability $\truthprob{\agent{y}}$ only (the communications are stateless): the two possible answers of $\agent{y}$ are, of course,   $\sharedbelief = \truesymbol$ and $\sharedbelief = \falsesymbol$.

Agent $\agent{x}$ counts the number of exchanges where $\agent{y}$ answered truthfully ($\#_{\truesymbol}$) and when it lied ($\#_{\falsesymbol}$). Clearly, $\lboot = \#_{\truesymbol} + \#_{\falsesymbol}$. 
Using this evidence, $\agent{x}$ can form an opinion on $\agent{y}$'s trustworthiness 
\[
\displaystyle{\opinionag{\agent{x}}{\agent{y}} = \tuple{\frac{\#_{\truesymbol}}{\lboot+2}, \frac{\#_{\falsesymbol}}{\lboot + 2}, \frac{2}{\lboot + 2}}}
\]

\noindent
which should be close (according to the definition of distance given in Def. \ref{def:distance})  to the ``ideal'' (``real'') opinion the (omniscient) experimenter has on $\agent{y}$, \viz{} 

\[
\displaystyle{\opinionag{Exp}{\agent{y}}} = \tuple{\truthprob{\agent{y}}, (1-\truthprob{\agent{y}}), 0}
\]

Therefore, during the bootstrapping phase, each agent $\agent{y} \in \neighbours{\agent{x}}$, $\agent{x}$ records its opinion of $\agent{y}$ in its knowledge base.

In the experiment described in \cite{Cerutti2013a}, $\lboot$ varies between 25 and 250: however, as we noted in that paper, this variation did not alter the experiment's results. Therefore, in this paper we collected data for ten different values of $\lboot$ varying it between 2 and 29 with a step size of 3.

\subsection{Phase II: Exploration} 
\label{sec:exp-exploring-network}

After each agent has enriched its knowledge base with opinions of its connections' trustworthiness, an ``explorer'' $\theagent \in \setagents$ is randomly selected. The task of this explorer is to determine the trustworthiness of each agent in the network. The explorer, $\theagent$, acquires information about the network by asking its connections ``Who are your connections?''. Each agent $\agent{y} \in \neighbours{\theagent}$ answers this question according to $\truthprob{\agent{y}}$, which means that their answers are: for each $\agent{y} \in \neighbours{\agent{x}}$ $Connections_{\agent{y}} \subseteq  \neighbours{\agent{y}}$ (clearly if $\truthprob{\agent{y}} = 1$, then \linebreak  $Connections_{\agent{y}} = \neighbours{\agent{y}}$).

Agent $\theagent$ collects all the answers and creates a set of tuples associating agents that the explorer does not directly know, and all the agents that have revealed that they have connections to that agent, viz.:

\[
\begin{array}{l l}
  \displaystyle{\mathcal{M}} = & \displaystyle{\{\tuple{\agent{z}, \set{\agent{y_1}, \ldots, \agent{y_n}}} ~|~ \forall i \in [1\ldots n]~ \agent{y_i} \in \neighbours{\theagent}} \\
    & \displaystyle{~~\mbox{ and }  \allowbreak \agent{z} \in \bigcap_{i = 1}^n \neighbours{\agent{y_i}}\}}
\end{array}
\]

Then, for each pair of $\mathcal{M}$, $\tuple{\agent{z}, \set{\agent{y_1}, \ldots, \agent{y_n}}}$, 
such that $\agent{z} \notin \neighbours{\theagent} \cup \set{\theagent}$ (i.e. for each agent it is not directly connected to), $\theagent$ asks each $\agent{y_i}$ (i.e. those that are connected to that agent) about $\opinionag{\agent{y_i}}{\agent{z}}$ (i.e. their opinion of that agent). $\agent{y_i}$ answers according to $\truthprob{\agent{y_i}}$ either $\opinionag{\agent{y_i}}{\agent{z}}$ or $\opinion{R}$ where $\opinion{R}$ is a SL opinion computed randomly such that $\opinion{R} \neq \opinionag{\agent{y_i}}{\agent{z}}$. Since $\theagent$ cannot determine whether the answer is true or not,  we abuse notation,  by associating  $\opinionag{\agent{y_i}}{\agent{z}}$ with the answer $\theagent$ received from $\agent{y_i}$ to the question ``What is your opinion about $\agent{z}$?''.

Subsequently, $\theagent$ computes ${\opinionag{\theagent}{\agent{z}}}_{|J} = (\opinionag{\theagent}{\agent{y_1}} \josangdiscountsym \opinionag{\agent{y_1}}{\agent{z}}) \josangfusionsym \ldots \josangfusionsym (\opinionag{\theagent}{\agent{y_n}} \josangdiscountsym \opinionag{\agent{y_n}}{\agent{z}})$ (\viz{} the fusion of the discounted opinions on $\agent{z}$ of its connections using J{\o}sang's operators), and, for each operator defined in Def. \ref{def:naive} and Def. \ref{def:graphicaldiscountoperators} (\ie{} $\circ \in \set{\optrustconf{}{}{n}, \optrustconf{}{}{1}, \optrustconf{}{}{2}, \optrustconf{}{}{3}}$ )
${\opinionag{\theagent}{\agent{z}}}_{|\circ} = \fusionop{1}((\opinionag{\theagent}{\agent{y_1}} \circ \opinionag{\agent{y_1}}{\agent{z}}), \ldots, (\opinionag{\theagent}{\agent{y_n}} \circ \opinionag{\agent{y_n}}{\agent{z}})$ (\viz{} the fusion of the discounted opinions on $\agent{z}$ of its connections using the na\"ive operator, and the members of the graphical operator family introduced in Def. \ref{defn:combination} with the fusion operator of Def. \ref{def:graphicalfusion}). Since we considered only one fusion operator, viz. \fusionop{1}, we will write it without a subscript to improve readability. Moreover, since we want to evaluate the proposed operators and compare them to J{\o}sang's ones, each exploration has been performed with J{\o}sang's operators and with only one of the proposed operators. We therefore explore the same network four times with J{\o}sang's operator. However, doing so guarantees that the evaluation of each operator is independent from other evaluations.

Finally, each agent $\agent{z}$ is added to the list of the connections of $\theagent$ and the process starts again by setting $\mathcal{M} = \emptyset$ and querying each member of the connections until, in two subsequent interactions, no further agents are added to $\theagent$'s connections. This exploration process therefore enables \theagent{} to form a picture of the agents in the network that does not have a direct link to through the opinions of other agents. Note that the results obtained through iteration $j$ of the exploration phase serves to bootstrap iteration $j + 1$.

\subsection{Computing the Distances.}

\begin{figure*}[htb]
  \centering
  \begin{subfigure}[b]{0.30\textwidth}
    \resizebox{\textwidth}{!}{
%
%
\begin{psfrags}%
\psfragscanon%
%
\psfrag{s10}[][]{\color[rgb]{0,0,0}\setlength{\tabcolsep}{0pt}\begin{tabular}{c} \end{tabular}}%
\psfrag{s11}[][]{\color[rgb]{0,0,0}\setlength{\tabcolsep}{0pt}\begin{tabular}{c} \end{tabular}}%
\psfrag{s12}[l][l]{\color[rgb]{0,0,0}Gamma function}%
\psfrag{s13}[l][l]{\color[rgb]{0,0,0}histogram}%
\psfrag{s14}[l][l]{\color[rgb]{0,0,0}Gamma function}%
%
\psfrag{x01}[t][t]{0}%
\psfrag{x02}[t][t]{0.2}%
\psfrag{x03}[t][t]{0.4}%
\psfrag{x04}[t][t]{0.6}%
\psfrag{x05}[t][t]{0.8}%
\psfrag{x06}[t][t]{1}%
\psfrag{x07}[t][t]{1.2}%
\psfrag{x08}[t][t]{1.4}%
%
\psfrag{v01}[r][r]{0}%
\psfrag{v02}[r][r]{5000}%
\psfrag{v03}[r][r]{10000}%
\psfrag{v04}[r][r]{15000}%
\psfrag{v05}[r][r]{20000}%
\psfrag{v06}[r][r]{25000}%
\psfrag{ypower2}[Bl][Bl]{$\times 10^{4}$}%
%
\resizebox{12cm}{!}{\includegraphics{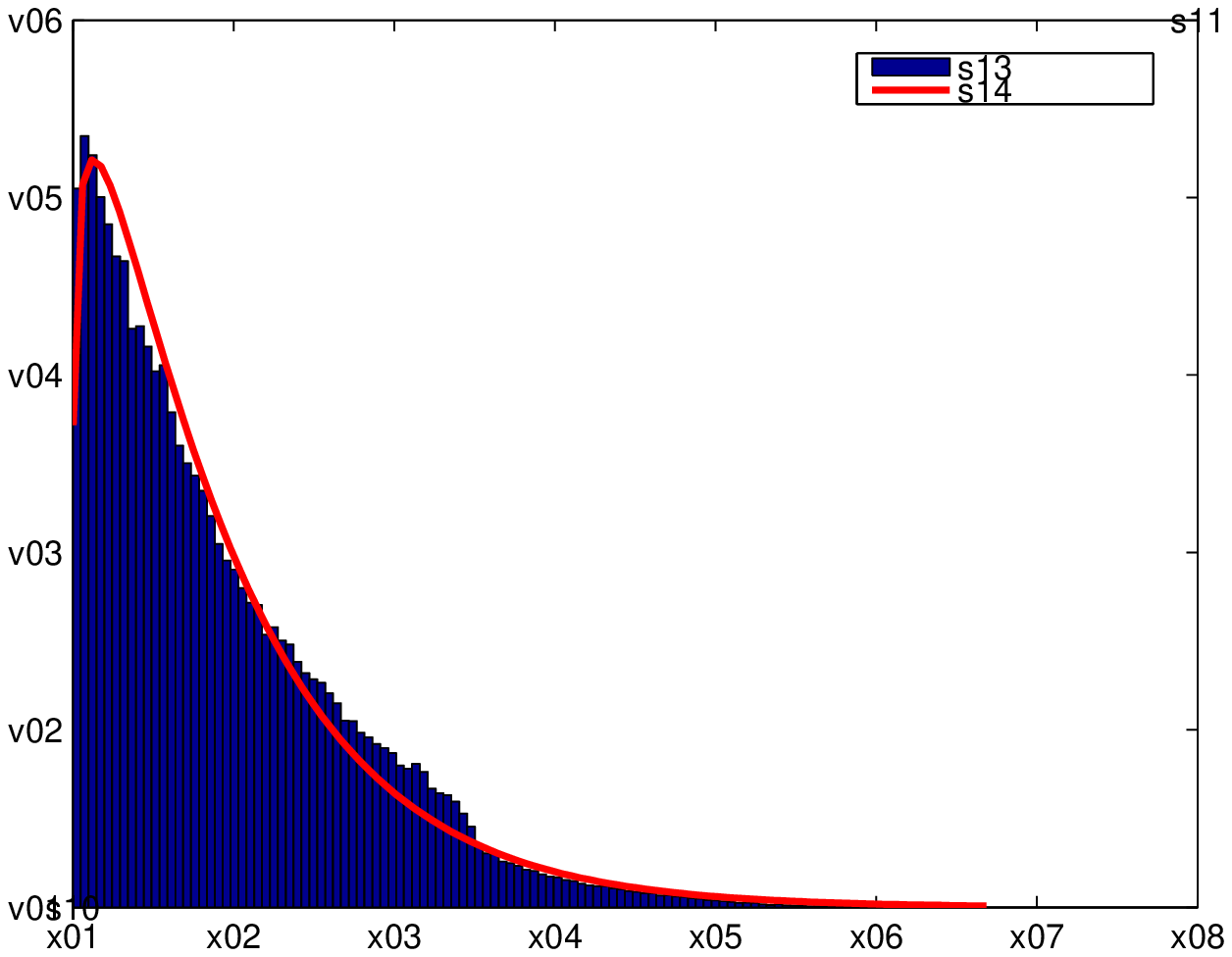}}%
\end{psfrags}%
%
}
    \caption{}
    \label{fig:distribution-expected-josang}
  \end{subfigure}
  \qquad
  \begin{subfigure}[b]{0.30\textwidth}
    \resizebox{\textwidth}{!}{
%
%
\begin{psfrags}%
\psfragscanon%
%
\psfrag{s10}[][]{\color[rgb]{0,0,0}\setlength{\tabcolsep}{0pt}\begin{tabular}{c} \end{tabular}}%
\psfrag{s11}[][]{\color[rgb]{0,0,0}\setlength{\tabcolsep}{0pt}\begin{tabular}{c} \end{tabular}}%
\psfrag{s12}[l][l]{\color[rgb]{0,0,0}Gamma function}%
\psfrag{s13}[l][l]{\color[rgb]{0,0,0}histogram}%
\psfrag{s14}[l][l]{\color[rgb]{0,0,0}Gamma function}%
%
\psfrag{x01}[t][t]{0}%
\psfrag{x02}[t][t]{0.5}%
\psfrag{x03}[t][t]{1}%
\psfrag{x04}[t][t]{1.5}%
\psfrag{x05}[t][t]{2}%
\psfrag{x06}[t][t]{2.5}%
\psfrag{x07}[t][t]{3}%
%
\psfrag{v01}[r][r]{0}%
\psfrag{v02}[r][r]{2000}%
\psfrag{v03}[r][r]{4000}%
\psfrag{v04}[r][r]{6000}%
\psfrag{v05}[r][r]{8000}%
\psfrag{v06}[r][r]{10000}%
\psfrag{v07}[r][r]{12000}%
%
\resizebox{12cm}{!}{\includegraphics{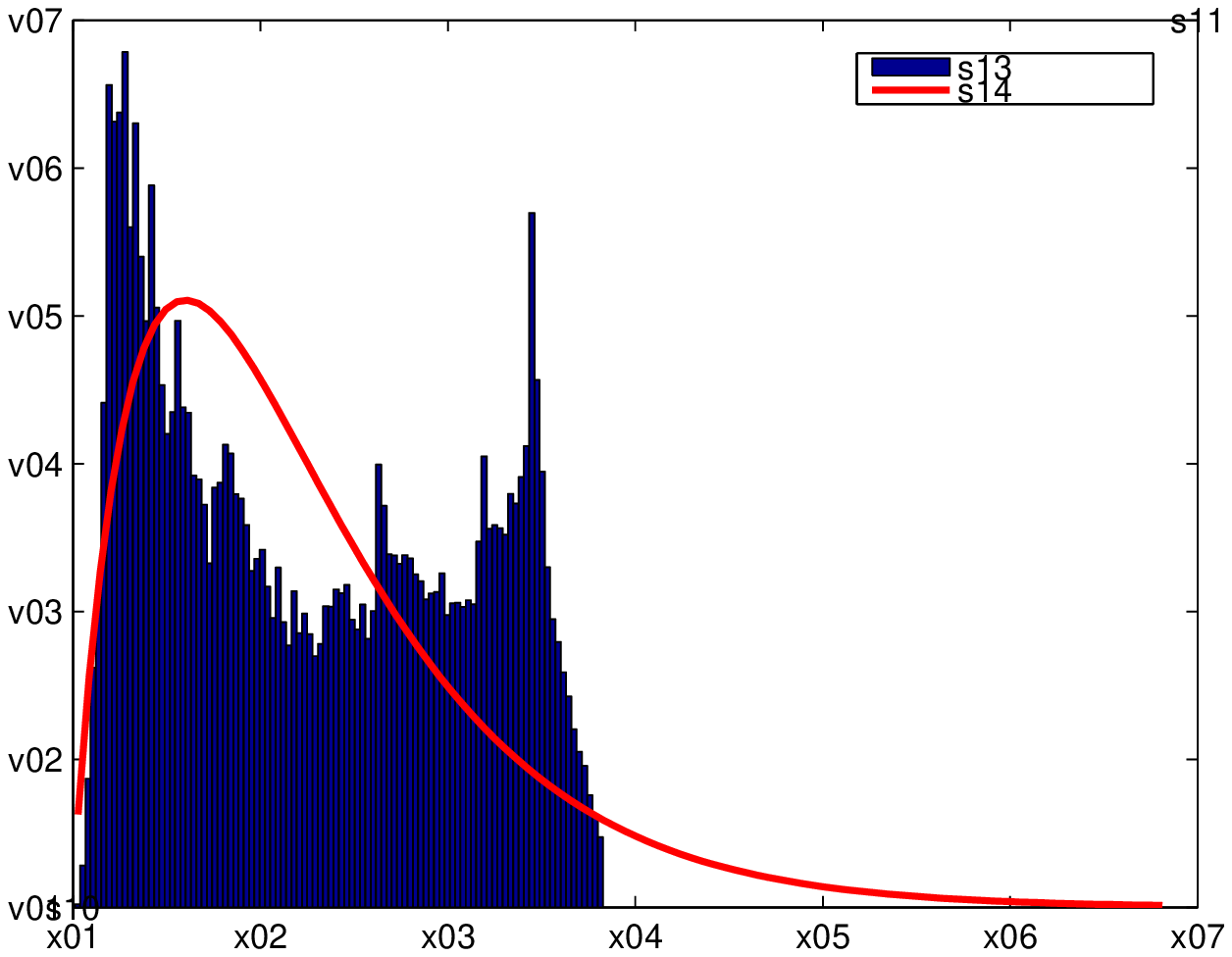}}%
\end{psfrags}%
%
}
    \caption{}
    \label{fig:distribution-geometric-josang}
  \end{subfigure}
  \caption{Histogram and Gamma function fitting of $\distanceexp{{\opinionag{\theagent}{\agent{z}}}_{|J}}{\opinionag{Exp}{\agent{z}}}$, Fig. (\subref{fig:distribution-expected-josang}), and of $\distance{{\opinionag{\theagent}{\agent{z}}}_{|J}}{\opinionag{Exp}{\agent{z}}}$, Fig. (\subref{fig:distribution-geometric-josang}).}
  \label{fig:distribution-josang}
\end{figure*}

For each agent $\agent{z} \in \setagents \setminus \set{\theagent}$, for each $\circ \in \set{\optrustconf{}{}{n}, \optrustconf{}{}{1}, \optrustconf{}{}{2}, \optrustconf{}{}{3}}$, we compute the distance between the two derived opinions ${\opinionag{\theagent}{\agent{z}}}_{|J}$ and ${\opinionag{\theagent}{\agent{z}}}_{|\circ}$, and the ``ideal'' opinion $\opinionag{Exp}{\agent{z}}$. For this purpose, we consider two notions of distance, namely a \emph{geometrical distance}, and a \emph{distance of expected values} \citep{Kaplan13personal}.

The \emph{geometric distance} between two opinions $\opinion{O_1}$ and $\opinion{O_2}$ is the Euclidean distance between the two point in the $\opinionset$ space.

\begin{definition}
  \label{def:distance}
  Given two opinions $O_1 = \opinion{O_1}$ and $O_2 = \opinion{O_2}$, the \emph{geometric distance} between $O_1$ and $O_2$ is 
  {\small
  \[
  \distance{O_1}{O_2} = \sqrt{(\belief{O_2} - \belief{O_1})^2 + (\disbelief{O_2} - \disbelief{O_1})^2 + (\uncertainty{O_2} - \uncertainty{O_1})^2 }
  \]
  }
\end{definition}

Another interesting distance measure is difference between the \emph{expected values} of subjective logic opinions. Let us recall that the expected value for a subjective logic opinion $\tuple{b_X, d_X, u_x, a_x}$ is $b_x + u_x \cdot a_x$. Given that  we assume a fixed base rate of $\frac{1}{2}$,  the expected value simplifies to $b_x + \frac{u_x}{2}$. Given the expected value of two opinions, we can easily compute the distance between them.

\begin{definition}
  \label{def:distanceexp}
  Given two opinions $O_1 = \opinion{O_1}$ and $O_2 = \opinion{O_2}$, the \emph{expected value distance} between $O_1$ and $O_2$ is 
  \[
  \displaystyle{\distanceexp{O_1}{O_2} = \left| \left(\belief{O_1} + \frac{\uncertainty{O_1}}{2} \right) - \left(\belief{O_2} + \frac{\uncertainty{O_2}}{2}\right)\right|}
  \]
\end{definition}

In other words, 
$\forall \agent{z} \in \setagents \setminus \set{\theagent}$, $\distance{{\opinionag{\theagent}{\agent{z}}}_{|J}}{\opinionag{Exp}{\agent{z}}}$ (reps. $\distanceexp{{\opinionag{\theagent}{\agent{z}}}_{|J}}{\opinionag{Exp}{\agent{z}}}$) is the geometric distance (resp. expected value distance)  between the derived opinion using J{\o}sang's operators and the ``ideal'' one (abbrev. $d^J_G$, resp. $d^J_E$), and, $\forall \circ \in \set{\optrustconf{}{}{n}, \optrustconf{}{}{1}, \optrustconf{}{}{2}, \optrustconf{}{}{3}}$, $\distance{{\opinionag{\theagent}{\agent{z}}}_{|\circ}}{\opinionag{Exp}{\agent{z}}}$ (resp. $\distanceexp{{\opinionag{\theagent}{\agent{z}}}_{|\circ}}{\opinionag{Exp}{\agent{z}}}$) is the geometric distance (resp. expected value distance) between the derived opinion using either the na\"ive operators or the operators of the family of graphical discount operators (Def. \ref{defn:combination}) and the ``ideal'' one (abbrev. $d^{\circ}_G$, resp. $d^{\circ}_E$). 

Finally, for each $\agent{z} \in \setagents \setminus \set{\theagent}$, $\forall \circ \in \set{\optrustconf{}{}{n}, \optrustconf{}{}{1}, \optrustconf{}{}{2}, \optrustconf{}{}{3}}$, we compare the two computed distances obtaining the following scalar comparison values:

\[
r_{G}(\agent{z}) = 
\begin{cases}
 
\displaystyle{-\log{\frac{\distance{{\opinionag{\theagent}{\agent{z}}}_{|\circ}}{\opinionag{Exp}{\agent{z}}}}{\distance{{\opinionag{\theagent}{\agent{z}}}_{|J}}{\opinionag{Exp}{\agent{z}}}}}} & \begin{array}{l} \mbox{\hspace{-0.7em}if } 
\distance{{\opinionag{\theagent}{\agent{z}}}_{|\circ}}{\opinionag{Exp}{\agent{z}}} > \\
 ~~~~ \distance{{\opinionag{\theagent}{\agent{z}}}_{|J}}{\opinionag{Exp}{\agent{z}}} \\  
\end{array}\\

\mbox{} & \mbox{}\\
 \displaystyle{\log{\frac{\distance{{\opinionag{\theagent}{\agent{z}}}_{|J}}{\opinionag{Exp}{\agent{z}}}}{\distance{{\opinionag{\theagent}{\agent{z}}}_{|\circ}}{\opinionag{Exp}{\agent{z}}}}}} & \begin{array}{l} \mbox{\hspace{-0.7em}if }  \distance{{\opinionag{\theagent}{\agent{z}}}_{|J}}{\opinionag{Exp}{\agent{z}}} \geq \\
 ~~~~ \distance{{\opinionag{\theagent}{\agent{z}}}_{|\circ}}{\opinionag{Exp}{\agent{z}}}
\end{array}\\
\end{cases}
\]

\noindent
and

\[
r_{E}(\agent{z}) = 
\begin{cases}
 
\displaystyle{-\log{\frac{\distanceexp{{\opinionag{\theagent}{\agent{z}}}_{|\circ}}{\opinionag{Exp}{\agent{z}}}}{\distanceexp{{\opinionag{\theagent}{\agent{z}}}_{|J}}{\opinionag{Exp}{\agent{z}}}}}} & \begin{array}{l} \mbox{\hspace{-0.7em}if } 
\distanceexp{{\opinionag{\theagent}{\agent{z}}}_{|\circ}}{\opinionag{Exp}{\agent{z}}} > \\
 ~~~~ \distanceexp{{\opinionag{\theagent}{\agent{z}}}_{|J}}{\opinionag{Exp}{\agent{z}}} \\  
\end{array}\\

\mbox{} & \mbox{}\\
 \displaystyle{\log{\frac{\distanceexp{{\opinionag{\theagent}{\agent{z}}}_{|J}}{\opinionag{Exp}{\agent{z}}}}{\distanceexp{{\opinionag{\theagent}{\agent{z}}}_{|\circ}}{\opinionag{Exp}{\agent{z}}}}}} & \begin{array}{l} \mbox{\hspace{-0.7em}if }  \distanceexp{{\opinionag{\theagent}{\agent{z}}}_{|J}}{\opinionag{Exp}{\agent{z}}} \geq \\
 ~~~~ \distanceexp{{\opinionag{\theagent}{\agent{z}}}_{|\circ}}{\opinionag{Exp}{\agent{z}}}
\end{array}\\
\end{cases}
\]

To strengthen the significance of the results, we mainly concentrate on averages, and thus  $\theagent$ explores the network $|\setagents|/2  = 25$ times; we write $r_G(\agent{z})$ (resp. $r_E(\agent{z})$) to denote the average of the 25 computed logarithmic ratios using the geometric distance (resp. expected value distance). Moreover $\overline{r_G(\agent{z})} = average_{\agent{z} \in \setagents \setminus \set{\theagent}} r(\agent{z})$ (resp. $\overline{r_E(\agent{z})} = \linebreak average_{\agent{z} \in \setagents \setminus \set{\theagent}} r(\agent{z})$) is the average of the comparison value over the whole set of agents using the geometric distance (resp. expected value distance).

\section{Analysis of Experimental Results}
\label{sec:results}

To ensure that the outcomes are  not biased by the random generator, we run the same experiment ten times. Each run follows the steps described in Sect. \ref{sec:experiment}, and thus for each value of $\linkprob$, 10 networks have been generated randomly. Moreover, since each agent can lie, each generated network has been explored 25 times. Therefore, for each run, for each value of $\linkprob$, 250 explorations over 10 different networks have been carried on (i.e.~12500 explorations were considered in this experiment).

In Section \ref{sec:distr-dist} we make a qualitative analyse of the distributions of the distances computed using the two metrics discussed in Definitions \ref{def:distance} and \ref{def:distanceexp}.  Section \ref{sec:analys-using-wilc} discusses the results of the Wilcoxon signed-rank test on the measure of distances. Finally, Sect. \ref{sec:dynamics-results} provides a qualitative analysis of the dynamics of the results varying the two parameters of the experiment, namely the probability of connections $\linkprob$ and the bootstrap time $\lboot$.

\begin{figure*}[p]
  \centering
  \begin{subfigure}[b]{0.30\textwidth}
    \resizebox{\textwidth}{!}{
%
%
\begin{psfrags}%
\psfragscanon%
%
\psfrag{s10}[][]{\color[rgb]{0,0,0}\setlength{\tabcolsep}{0pt}\begin{tabular}{c} \end{tabular}}%
\psfrag{s11}[][]{\color[rgb]{0,0,0}\setlength{\tabcolsep}{0pt}\begin{tabular}{c} \end{tabular}}%
\psfrag{s12}[l][l]{\color[rgb]{0,0,0}Gamma function}%
\psfrag{s13}[l][l]{\color[rgb]{0,0,0}histogram}%
\psfrag{s14}[l][l]{\color[rgb]{0,0,0}Gamma function}%
%
\psfrag{x01}[t][t]{0}%
\psfrag{x02}[t][t]{0.2}%
\psfrag{x03}[t][t]{0.4}%
\psfrag{x04}[t][t]{0.6}%
\psfrag{x05}[t][t]{0.8}%
\psfrag{x06}[t][t]{1}%
\psfrag{x07}[t][t]{1.2}%
\psfrag{x08}[t][t]{1.4}%
%
\psfrag{v01}[r][r]{0}%
\psfrag{v02}[r][r]{5000}%
\psfrag{v03}[r][r]{10000}%
\psfrag{v04}[r][r]{15000}%
\psfrag{v05}[r][r]{20000}%
\psfrag{v06}[r][r]{25000}%
\psfrag{ypower2}[Bl][Bl]{$\times 10^{4}$}%
%
\resizebox{12cm}{!}{\includegraphics{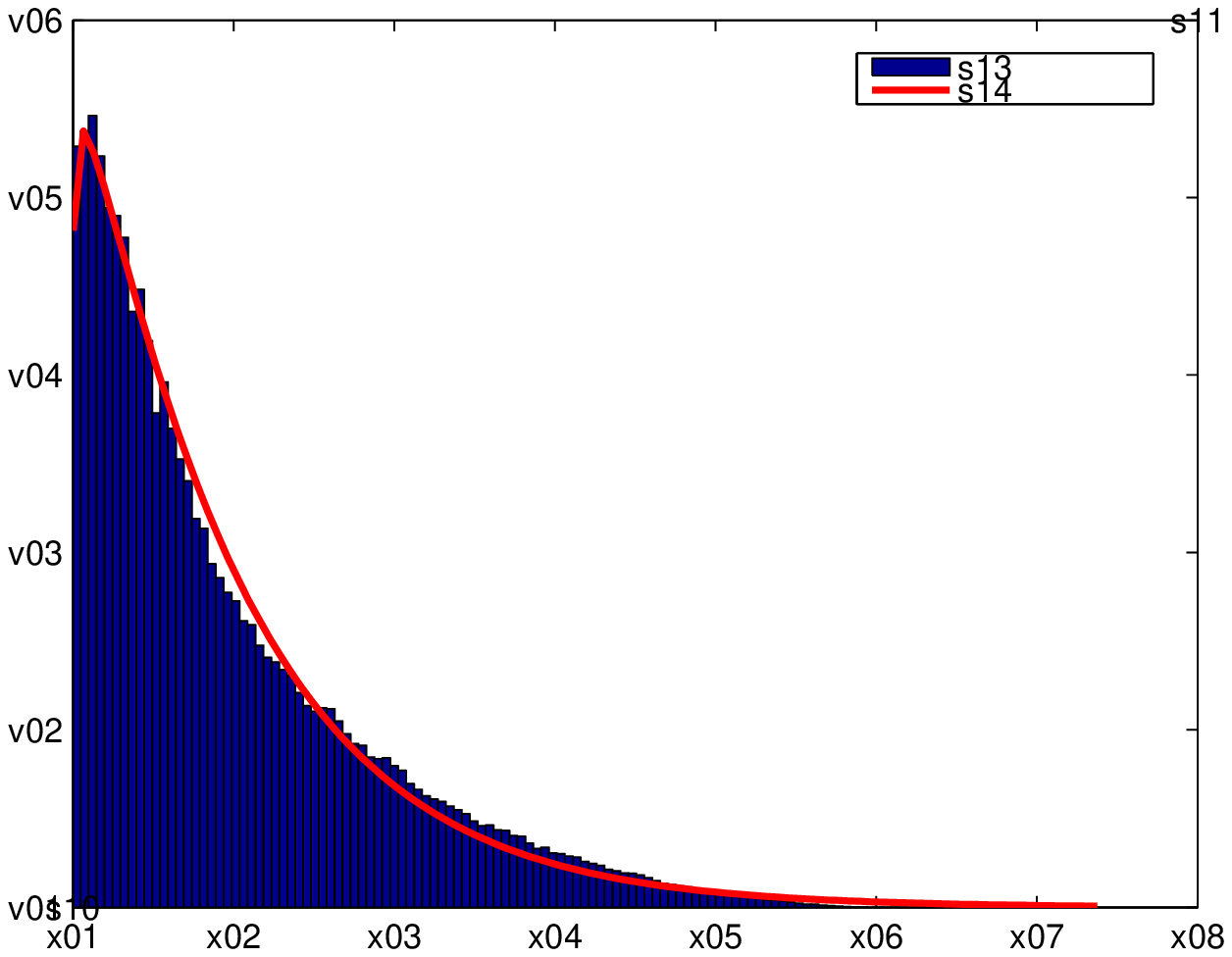}}%
\end{psfrags}%
%
}
    \caption{}
    \label{fig:distribution-expected-traditional}
  \end{subfigure}
  \qquad
  \begin{subfigure}[b]{0.30\textwidth}
    \resizebox{\textwidth}{!}{
%
%
\begin{psfrags}%
\psfragscanon%
%
\psfrag{s10}[][]{\color[rgb]{0,0,0}\setlength{\tabcolsep}{0pt}\begin{tabular}{c} \end{tabular}}%
\psfrag{s11}[][]{\color[rgb]{0,0,0}\setlength{\tabcolsep}{0pt}\begin{tabular}{c} \end{tabular}}%
\psfrag{s12}[l][l]{\color[rgb]{0,0,0}Gamma function}%
\psfrag{s13}[l][l]{\color[rgb]{0,0,0}histogram}%
\psfrag{s14}[l][l]{\color[rgb]{0,0,0}Gamma function}%
%
\psfrag{x01}[t][t]{0}%
\psfrag{x02}[t][t]{0.5}%
\psfrag{x03}[t][t]{1}%
\psfrag{x04}[t][t]{1.5}%
\psfrag{x05}[t][t]{2}%
\psfrag{x06}[t][t]{2.5}%
%
\psfrag{v01}[r][r]{0}%
\psfrag{v02}[r][r]{5000}%
\psfrag{v03}[r][r]{10000}%
\psfrag{v04}[r][r]{15000}%
%
\resizebox{12cm}{!}{\includegraphics{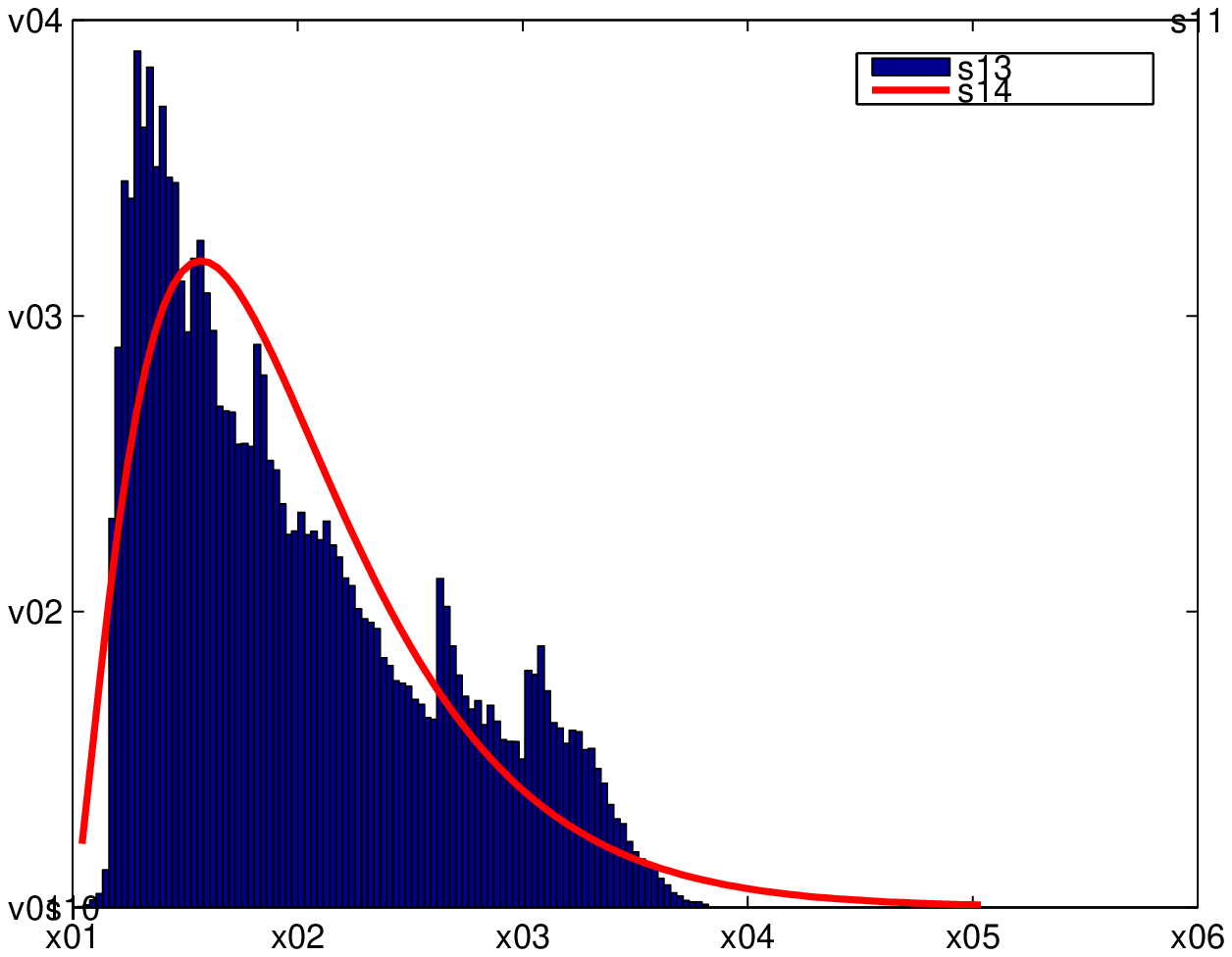}}%
\end{psfrags}%
%
}
    \caption{}
    \label{fig:distribution-geometric-traditional}
  \end{subfigure}

  \begin{subfigure}[b]{0.30\textwidth}
    \resizebox{\textwidth}{!}{
%
%
\begin{psfrags}%
\psfragscanon%
%
\psfrag{s10}[][]{\color[rgb]{0,0,0}\setlength{\tabcolsep}{0pt}\begin{tabular}{c} \end{tabular}}%
\psfrag{s11}[][]{\color[rgb]{0,0,0}\setlength{\tabcolsep}{0pt}\begin{tabular}{c} \end{tabular}}%
\psfrag{s12}[l][l]{\color[rgb]{0,0,0}Gamma function}%
\psfrag{s13}[l][l]{\color[rgb]{0,0,0}histogram}%
\psfrag{s14}[l][l]{\color[rgb]{0,0,0}Gamma function}%
%
\psfrag{x01}[t][t]{0}%
\psfrag{x02}[t][t]{0.2}%
\psfrag{x03}[t][t]{0.4}%
\psfrag{x04}[t][t]{0.6}%
\psfrag{x05}[t][t]{0.8}%
\psfrag{x06}[t][t]{1}%
\psfrag{x07}[t][t]{1.2}%
\psfrag{x08}[t][t]{1.4}%
%
\psfrag{v01}[r][r]{0}%
\psfrag{v02}[r][r]{5000}%
\psfrag{v03}[r][r]{10000}%
\psfrag{v04}[r][r]{15000}%
\psfrag{v05}[r][r]{20000}%
\psfrag{v06}[r][r]{25000}%
\psfrag{ypower2}[Bl][Bl]{$\times 10^{4}$}%
%
\resizebox{12cm}{!}{\includegraphics{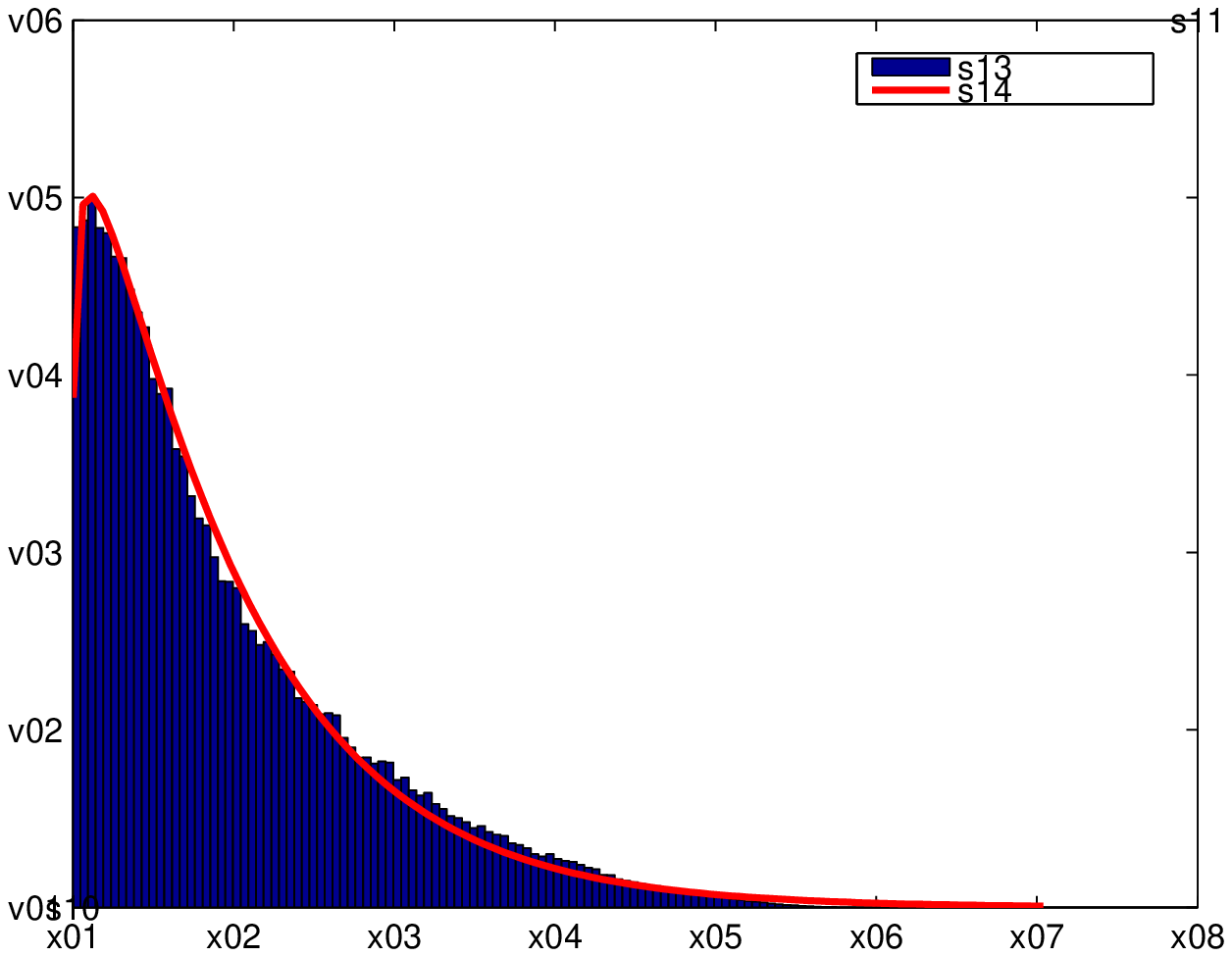}}%
\end{psfrags}%
%
}
    \caption{}
    \label{fig:distribution-expected-parallel}
  \end{subfigure}
  \qquad
  \begin{subfigure}[b]{0.30\textwidth}
    \resizebox{\textwidth}{!}{
%
%
\begin{psfrags}%
\psfragscanon%
%
\psfrag{s10}[][]{\color[rgb]{0,0,0}\setlength{\tabcolsep}{0pt}\begin{tabular}{c} \end{tabular}}%
\psfrag{s11}[][]{\color[rgb]{0,0,0}\setlength{\tabcolsep}{0pt}\begin{tabular}{c} \end{tabular}}%
\psfrag{s12}[l][l]{\color[rgb]{0,0,0}Gamma function}%
\psfrag{s13}[l][l]{\color[rgb]{0,0,0}histogram}%
\psfrag{s14}[l][l]{\color[rgb]{0,0,0}Gamma function}%
%
\psfrag{x01}[t][t]{0}%
\psfrag{x02}[t][t]{0.5}%
\psfrag{x03}[t][t]{1}%
\psfrag{x04}[t][t]{1.5}%
\psfrag{x05}[t][t]{2}%
\psfrag{x06}[t][t]{2.5}%
%
\psfrag{v01}[r][r]{0}%
\psfrag{v02}[r][r]{2000}%
\psfrag{v03}[r][r]{4000}%
\psfrag{v04}[r][r]{6000}%
\psfrag{v05}[r][r]{8000}%
\psfrag{v06}[r][r]{10000}%
\psfrag{v07}[r][r]{12000}%
\psfrag{v08}[r][r]{14000}%
%
\resizebox{12cm}{!}{\includegraphics{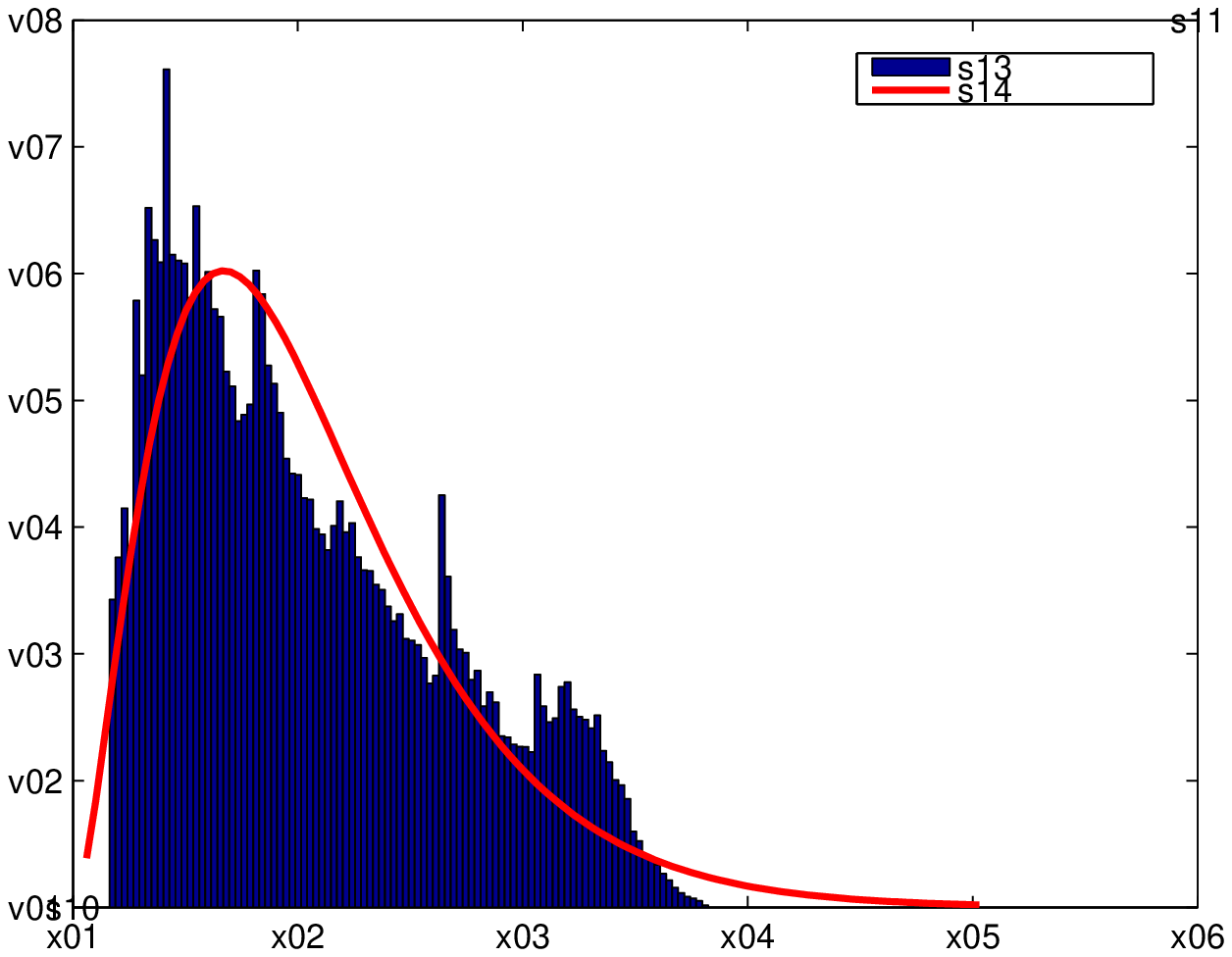}}%
\end{psfrags}%
%
}
    \caption{}
    \label{fig:distribution-geometric-parallel}
  \end{subfigure}

  \begin{subfigure}[b]{0.30\textwidth}
    \resizebox{\textwidth}{!}{
%
%
\begin{psfrags}%
\psfragscanon%
%
\psfrag{s10}[][]{\color[rgb]{0,0,0}\setlength{\tabcolsep}{0pt}\begin{tabular}{c} \end{tabular}}%
\psfrag{s11}[][]{\color[rgb]{0,0,0}\setlength{\tabcolsep}{0pt}\begin{tabular}{c} \end{tabular}}%
\psfrag{s12}[l][l]{\color[rgb]{0,0,0}Gamma function}%
\psfrag{s13}[l][l]{\color[rgb]{0,0,0}histogram}%
\psfrag{s14}[l][l]{\color[rgb]{0,0,0}Gamma function}%
%
\psfrag{x01}[t][t]{0}%
\psfrag{x02}[t][t]{0.2}%
\psfrag{x03}[t][t]{0.4}%
\psfrag{x04}[t][t]{0.6}%
\psfrag{x05}[t][t]{0.8}%
\psfrag{x06}[t][t]{1}%
\psfrag{x07}[t][t]{1.2}%
\psfrag{x08}[t][t]{1.4}%
%
\psfrag{v01}[r][r]{0}%
\psfrag{v02}[r][r]{2000}%
\psfrag{v03}[r][r]{4000}%
\psfrag{v04}[r][r]{6000}%
\psfrag{v05}[r][r]{8000}%
\psfrag{v06}[r][r]{10000}%
\psfrag{v07}[r][r]{12000}%
\psfrag{v08}[r][r]{14000}%
\psfrag{v09}[r][r]{16000}%
\psfrag{v10}[r][r]{18000}%
\psfrag{v11}[r][r]{20000}%
\psfrag{ypower2}[Bl][Bl]{$\times 10^{4}$}%
%
\resizebox{12cm}{!}{\includegraphics{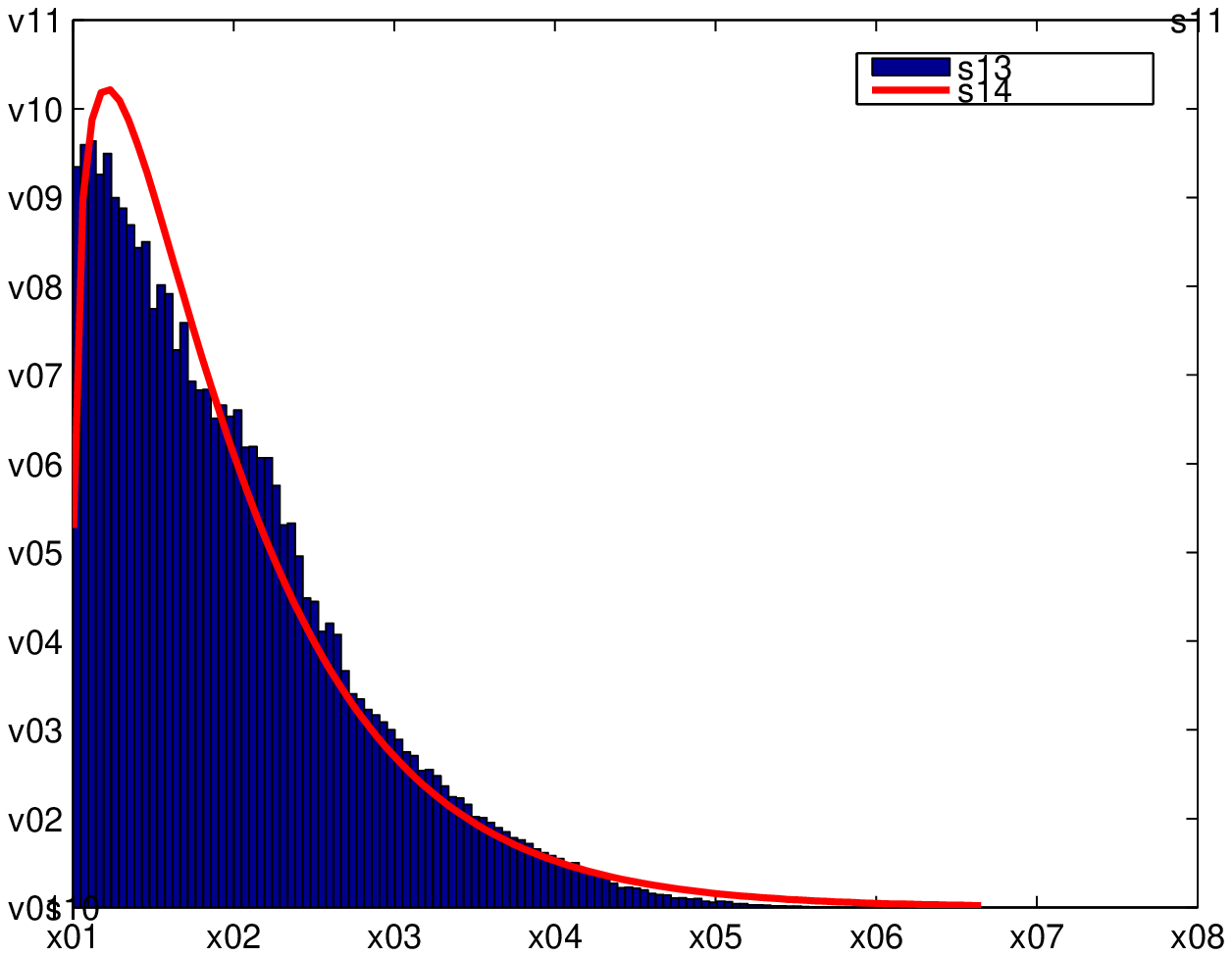}}%
\end{psfrags}%
%
}
    \caption{}
    \label{fig:distribution-expected-half}
  \end{subfigure}
  \qquad
  \begin{subfigure}[b]{0.30\textwidth}
    \resizebox{\textwidth}{!}{
%
%
\begin{psfrags}%
\psfragscanon%
%
\psfrag{s10}[][]{\color[rgb]{0,0,0}\setlength{\tabcolsep}{0pt}\begin{tabular}{c} \end{tabular}}%
\psfrag{s11}[][]{\color[rgb]{0,0,0}\setlength{\tabcolsep}{0pt}\begin{tabular}{c} \end{tabular}}%
\psfrag{s12}[l][l]{\color[rgb]{0,0,0}Gamma function}%
\psfrag{s13}[l][l]{\color[rgb]{0,0,0}histogram}%
\psfrag{s14}[l][l]{\color[rgb]{0,0,0}Gamma function}%
%
\psfrag{x01}[t][t]{0}%
\psfrag{x02}[t][t]{0.5}%
\psfrag{x03}[t][t]{1}%
\psfrag{x04}[t][t]{1.5}%
\psfrag{x05}[t][t]{2}%
\psfrag{x06}[t][t]{2.5}%
%
\psfrag{v01}[r][r]{0}%
\psfrag{v02}[r][r]{1000}%
\psfrag{v03}[r][r]{2000}%
\psfrag{v04}[r][r]{3000}%
\psfrag{v05}[r][r]{4000}%
\psfrag{v06}[r][r]{5000}%
\psfrag{v07}[r][r]{6000}%
\psfrag{v08}[r][r]{7000}%
\psfrag{v09}[r][r]{8000}%
\psfrag{v10}[r][r]{9000}%
\psfrag{v11}[r][r]{10000}%
%
\resizebox{12cm}{!}{\includegraphics{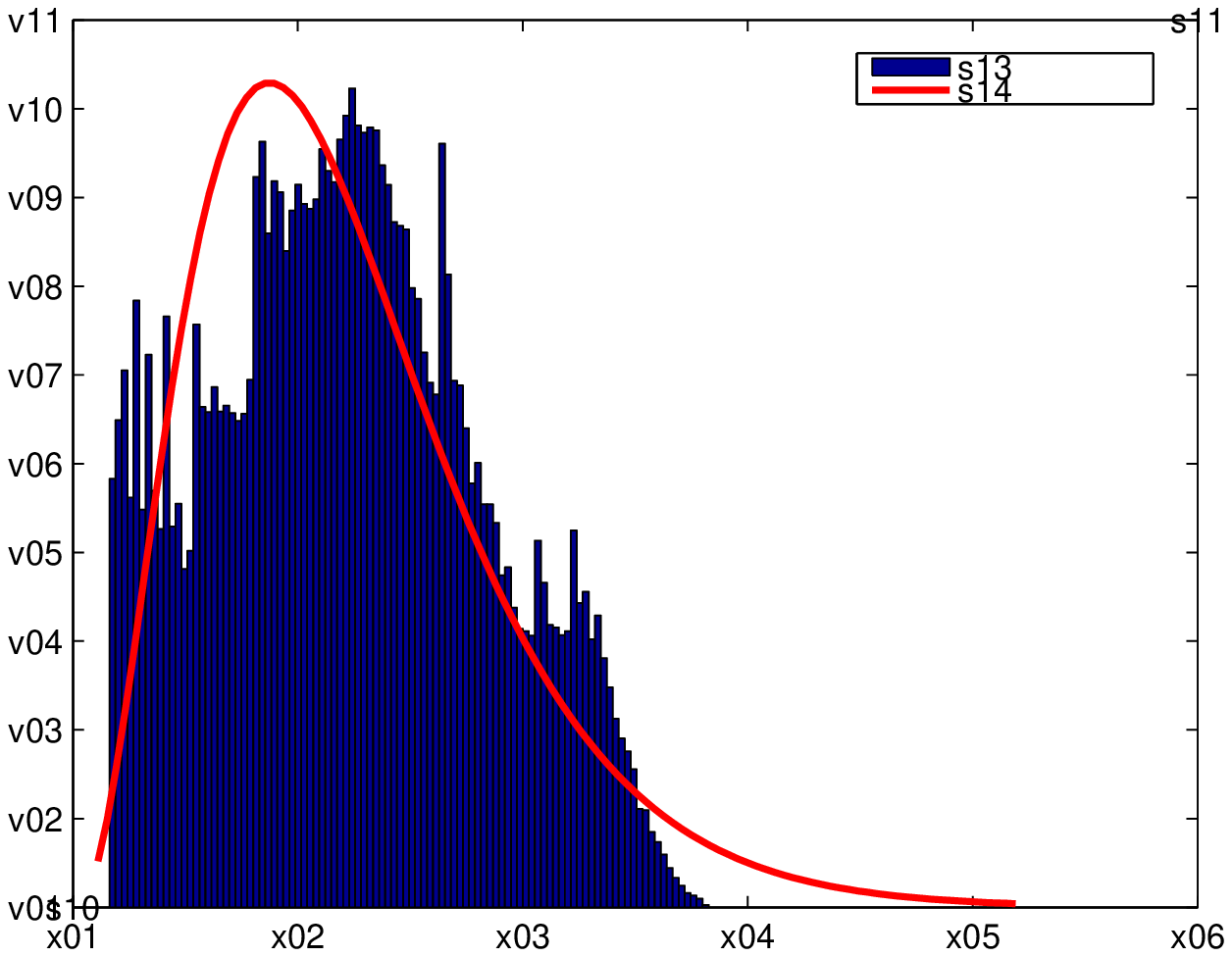}}%
\end{psfrags}%
%
}
    \caption{}
    \label{fig:distribution-geometric-half}
  \end{subfigure}

  \begin{subfigure}[b]{0.30\textwidth}
    \resizebox{\textwidth}{!}{
%
%
\begin{psfrags}%
\psfragscanon%
%
\psfrag{s10}[][]{\color[rgb]{0,0,0}\setlength{\tabcolsep}{0pt}\begin{tabular}{c} \end{tabular}}%
\psfrag{s11}[][]{\color[rgb]{0,0,0}\setlength{\tabcolsep}{0pt}\begin{tabular}{c} \end{tabular}}%
\psfrag{s12}[l][l]{\color[rgb]{0,0,0}Gamma function}%
\psfrag{s13}[l][l]{\color[rgb]{0,0,0}histogram}%
\psfrag{s14}[l][l]{\color[rgb]{0,0,0}Gamma function}%
%
\psfrag{x01}[t][t]{0}%
\psfrag{x02}[t][t]{0.2}%
\psfrag{x03}[t][t]{0.4}%
\psfrag{x04}[t][t]{0.6}%
\psfrag{x05}[t][t]{0.8}%
\psfrag{x06}[t][t]{1}%
\psfrag{x07}[t][t]{1.2}%
\psfrag{x08}[t][t]{1.4}%
%
\psfrag{v01}[r][r]{0}%
\psfrag{v02}[r][r]{5000}%
\psfrag{v03}[r][r]{10000}%
\psfrag{v04}[r][r]{15000}%
\psfrag{v05}[r][r]{20000}%
\psfrag{v06}[r][r]{25000}%
\psfrag{ypower2}[Bl][Bl]{$\times 10^{4}$}%
%
\resizebox{12cm}{!}{\includegraphics{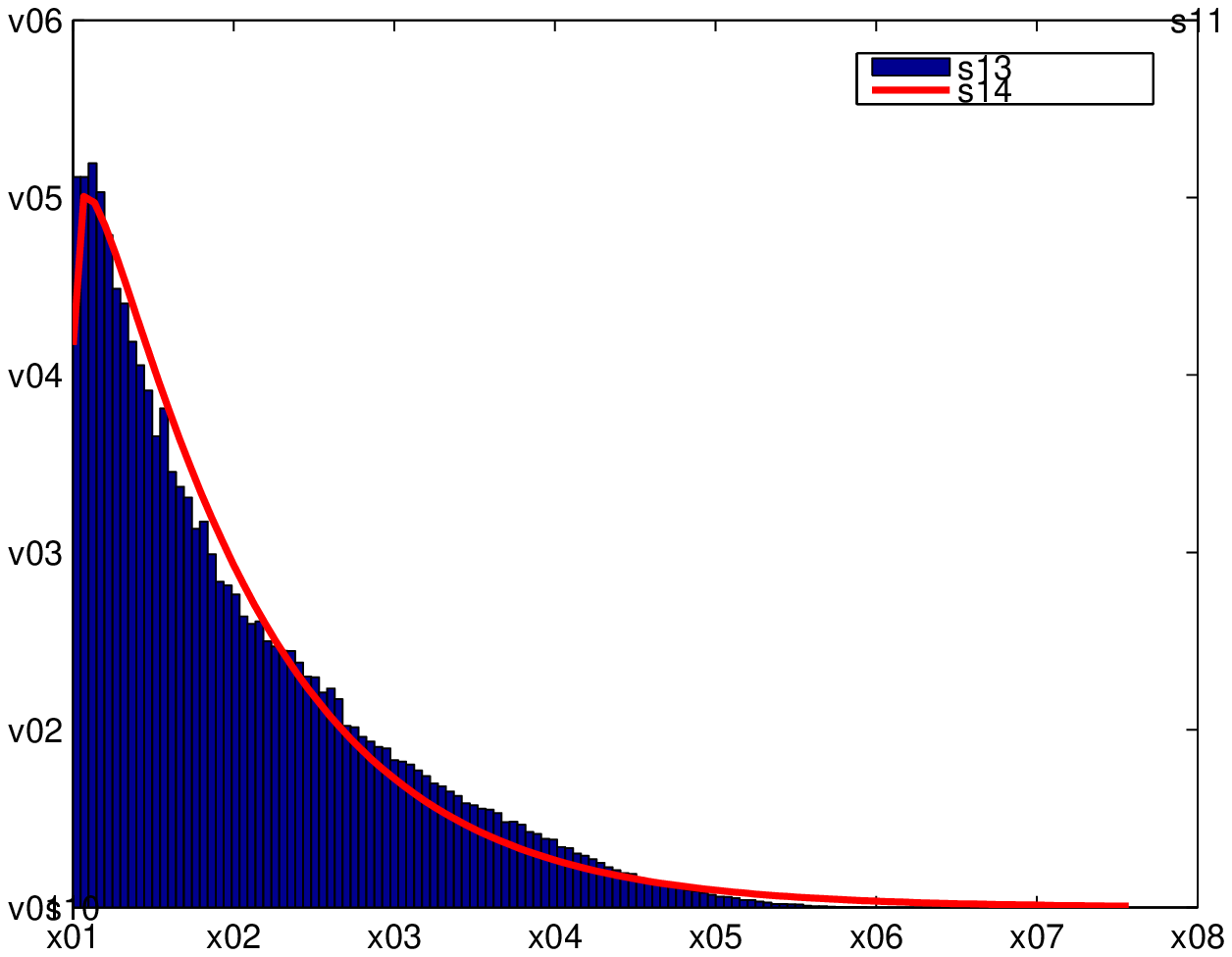}}%
\end{psfrags}%
%
}
    \caption{}
    \label{fig:distribution-expected-naive}
  \end{subfigure}
  \qquad
  \begin{subfigure}[b]{0.30\textwidth}
    \resizebox{\textwidth}{!}{
%
%
\begin{psfrags}%
\psfragscanon%
%
\psfrag{s10}[][]{\color[rgb]{0,0,0}\setlength{\tabcolsep}{0pt}\begin{tabular}{c} \end{tabular}}%
\psfrag{s11}[][]{\color[rgb]{0,0,0}\setlength{\tabcolsep}{0pt}\begin{tabular}{c} \end{tabular}}%
\psfrag{s12}[l][l]{\color[rgb]{0,0,0}Gamma function}%
\psfrag{s13}[l][l]{\color[rgb]{0,0,0}histogram}%
\psfrag{s14}[l][l]{\color[rgb]{0,0,0}Gamma function}%
%
\psfrag{x01}[t][t]{0}%
\psfrag{x02}[t][t]{0.5}%
\psfrag{x03}[t][t]{1}%
\psfrag{x04}[t][t]{1.5}%
\psfrag{x05}[t][t]{2}%
\psfrag{x06}[t][t]{2.5}%
%
\psfrag{v01}[r][r]{0}%
\psfrag{v02}[r][r]{2000}%
\psfrag{v03}[r][r]{4000}%
\psfrag{v04}[r][r]{6000}%
\psfrag{v05}[r][r]{8000}%
\psfrag{v06}[r][r]{10000}%
\psfrag{v07}[r][r]{12000}%
\psfrag{v08}[r][r]{14000}%
%
\resizebox{12cm}{!}{\includegraphics{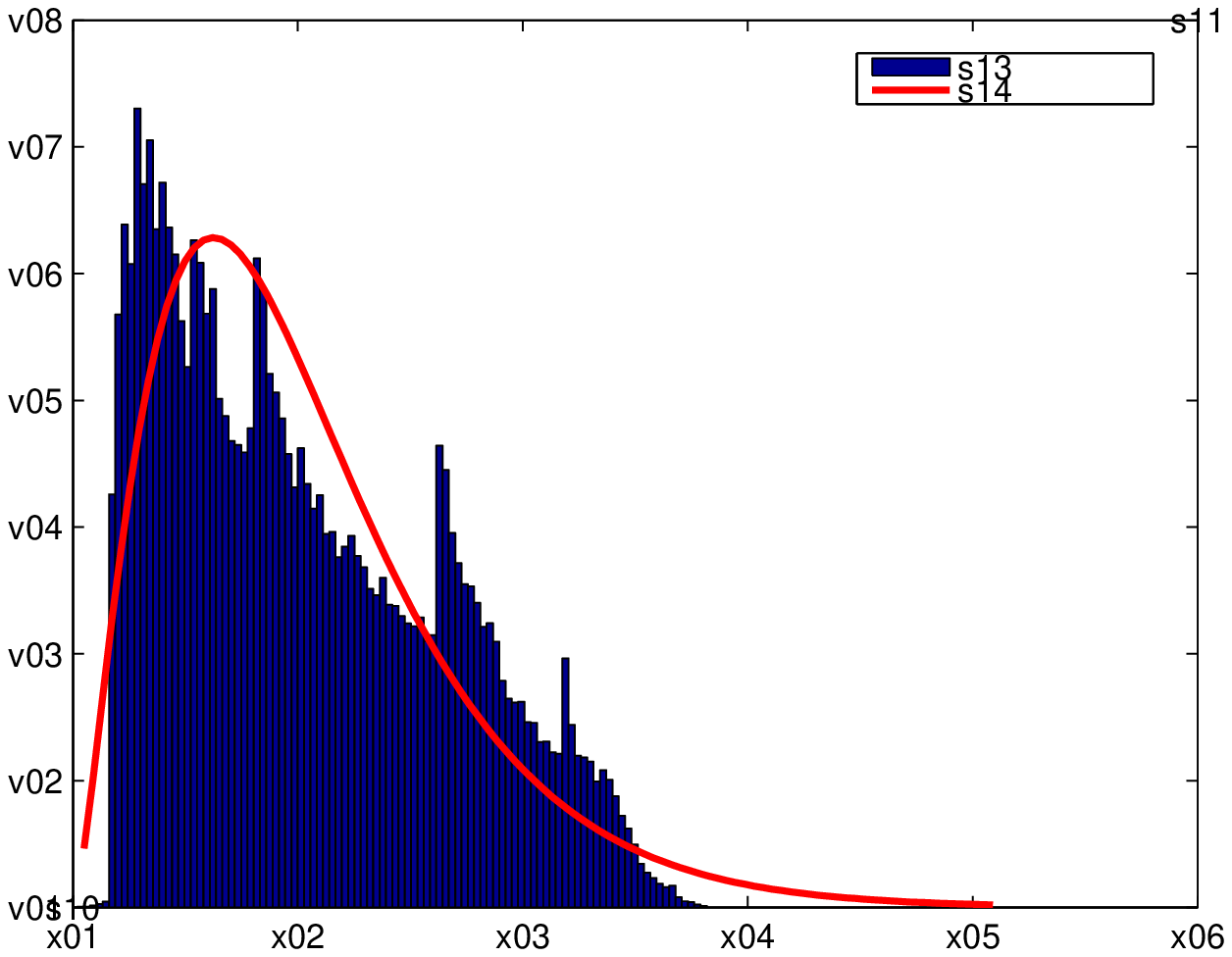}}%
\end{psfrags}%
%
}
    \caption{}
    \label{fig:distribution-geometric-naive}
  \end{subfigure}

  \caption{Histogram and Gamma function fitting of $\distanceexp{{\opinionag{\theagent}{\agent{z}}}_{|\optrustconf{}{}{}}}{\opinionag{Exp}{\agent{z}}}$ (resp.  $\distance{{\opinionag{\theagent}{\agent{z}}}_{|\optrustconf{}{}{}}}{\opinionag{Exp}{\agent{z}}}$), for $\optrustconf{}{}{} \in \set{\optrustconf{}{}{1}, \optrustconf{}{}{2}, \optrustconf{}{}{3}, \optrustconf{}{}{n}}$ Fig. (\subref{fig:distribution-expected-traditional}), (\subref{fig:distribution-expected-parallel}), (\subref{fig:distribution-expected-half}), (\subref{fig:distribution-expected-naive})  (Fig. (\subref{fig:distribution-geometric-traditional}), (\subref{fig:distribution-geometric-parallel}), (\subref{fig:distribution-geometric-half}), (\subref{fig:distribution-geometric-naive})).}
  \label{fig:distribution-ours-operators}
\end{figure*}

\subsection{Distributions of Distances}
\label{sec:distr-dist}

We first analysed the distributions of the distances between the ground truth and when using J{\o}sang's operators (Figures \ref{fig:distribution-expected-josang} and \ref{fig:distribution-geometric-josang} respectively);  
$\optrustconf{}{}{1}$ and $\fusionop{1}$ 
(Figures \ref{fig:distribution-expected-traditional} and \ref{fig:distribution-geometric-traditional});  
$\optrustconf{}{}{2}$ and $\fusionop{1}$
(Figures \ref{fig:distribution-expected-parallel} and \ref{fig:distribution-geometric-parallel}), $\optrustconf{}{}{3}$ and $\fusionop{1}$ 
(Figures \ref{fig:distribution-expected-half} and \ref{fig:distribution-geometric-half});  
$\optrustconf{}{}{n}$ and $\fusionop{1}$
(Figures \ref{fig:distribution-expected-naive} and \ref{fig:distribution-geometric-naive}).

Visual inspection of Figures \ref{fig:distribution-expected-josang}, \ref{fig:distribution-expected-traditional}, \ref{fig:distribution-expected-parallel}, \ref{fig:distribution-expected-half}, \ref{fig:distribution-expected-naive} indicates that the distances computed using the expected value distance (Definition \ref{def:distanceexp}) can be approximated (qualitatively\footnote{Figures
\ref{fig:distribution-expected-josang}, \ref{fig:distribution-geometric-josang}, 
\ref{fig:distribution-expected-traditional}, \ref{fig:distribution-geometric-traditional}, 
\ref{fig:distribution-expected-parallel}, \ref{fig:distribution-geometric-parallel}, 
\ref{fig:distribution-expected-half}, \ref{fig:distribution-geometric-half}, 
\ref{fig:distribution-expected-naive}, \ref{fig:distribution-geometric-naive}, are obtained using the Matlab function \texttt{histfit}.}) using a Gamma function, regardless of the choice of the operator. 

This result looks reasonable with respect to the opinions computed using J{\o}sang's operators (Fig. \ref{fig:distribution-expected-josang}) due to the fact that the experiment considered the Beta reputation system \cite{Ismail2002} and due to the statistical properties of J{\o}sang's operators (see \cite{Josang2006}). It is, however, interesting to note that the use of $\optrustconf{}{}{1}$ and $\fusionop{1}$, or $\optrustconf{}{}{2}$ and $\fusionop{1}$, or $\optrustconf{}{}{3}$ and $\fusionop{1}$, or $\optrustconf{}{}{n}$ and $\fusionop{1}$, all result in similar graphs.

More interesting is the fact that, considering the graphical distance (Figures \ref{fig:distribution-geometric-josang}, \ref{fig:distribution-geometric-traditional}, \ref{fig:distribution-geometric-parallel}, \ref{fig:distribution-geometric-half}, \ref{fig:distribution-geometric-naive}) we can conclude that (qualitatively) J{\o}sang's operators (and similarly $\optrustconf{}{}{3}$ with $\fusionop{1}$) are computing opinions whose geometric distance from the ground truth is not distributed using a Gamma function (Figures \ref{fig:distribution-geometric-josang} and \ref{fig:distribution-geometric-half}). On the other hand, using either $\optrustconf{}{}{1}$ with $\fusionop{1}$ (Fig. \ref{fig:distribution-geometric-traditional}), or $\optrustconf{}{}{2}$ with $\fusionop{1}$ (Fig. \ref{fig:distribution-geometric-parallel}), or $\optrustconf{}{}{n}$ with $\fusionop{1}$ (Fig. \ref{fig:distribution-geometric-naive}) returns opinions whose geometric distance from the ground truth has a interesting and regular shape, which are similar to a Gamma function or a Lognormal distribution. A comprehensive study of this is beyond the scope of the present paper and is left for future work.

\begin{table*}[htb]
  \renewcommand{\arraystretch}{1.5}
  \centering
  \begin{tabular}{|c|c|c|c|c|c|c|}\hline
    Operator & Md${}^{\star}$ $\distanceexp{{\opinionag{\theagent}{\agent{z}}}_{|\optrustconf{}{}{}}}{\opinionag{Exp}{\agent{z}}}$ & Md${}^{\star}$ $\distanceexp{{\opinionag{\theagent}{\agent{z}}}_{|J}}{\opinionag{Exp}{\agent{z}}}$ &  $s^-$ ($\times 10^{10}$) & $s^+$ ($\times 10^{10}$) & $z$ &  Incr. Performance${}^{\dagger}$\\\hline
    $\optrustconf{}{}{1}$ 	& $0.141$	& $0.144$	& $4.11$	& $4.53$	&  $-27.457^{\ddagger}$	&   $\approx~ + 5\%$ \\\hline
    \rowcolor{lightgray}
    $\optrustconf{}{}{n}$ 	& $0.156$	& $0.155$	& $4.40$	& $3.95$	&  $-29.586^{\ddagger}$  & $\approx - 5\%$\\\hline
    \rowcolor{lightgray}
    $\optrustconf{}{}{2}$	& $0.143$	& $0.142$	& $4.58$	& $3.89$	&  $-45.559^{\ddagger}$	& $\approx - 8\%$\\\hline
    \rowcolor{lightgray}
    $\optrustconf{}{}{3}$ 	& $0.163$	& $0.145$	& $5.12$	& $3.51$	&  $-104.098^{\ddagger}$ & $\approx - 19\%$\\\hline
    \multicolumn{7}{l}{${}^{\star}$ Median; \qquad ${}^{\dagger}$ computed as $(s^+ - s^-)/(s^+ + s^-)$; \qquad  ${}^{\ddagger} p < 0.001$.}
  \end{tabular}
  \caption{Wilcoxon signed-rank significance tests of distances derived using the expected value distance $\distanceexp{\cdot}{\cdot}$. In grey are the cases where the proposed operator $\optrustconf{}{}{} \in \set{\optrustconf{}{}{1}, \optrustconf{}{}{2}, \optrustconf{}{}{3}, \optrustconf{}{}{n}}$ did not outperform J{\o}sang's operators. Results are ordered by the increment of performance in descending order (i.e. the first row is the best one).\label{tablestata}}
\end{table*}

\begin{table*}[htb]
  \renewcommand{\arraystretch}{1.5}
  \centering
  \begin{tabular}{|c|c|c|c|c|c|c|}\hline
    Operator & Md${}^{\star}$ $\distance{{\opinionag{\theagent}{\agent{z}}}_{|\optrustconf{}{}{}}}{\opinionag{Exp}{\agent{z}}}$ & Md${}^{\star}$ $\distance{{\opinionag{\theagent}{\agent{z}}}_{|J}}{\opinionag{Exp}{\agent{z}}}$ &  $s^-$ ($\times 10^{10}$) & $s^+$ ($\times 10^{10}$) & $z$ &  Incr. Performance${}^{\dagger}$\\\hline
    $\optrustconf{}{}{1}$ 	& $0.415$	& $0.585$	& $1.91$	& $6.70$	& $-310.462^{\ddagger}$	& $\approx + 56\%$ \\\hline
    $\optrustconf{}{}{n}$ 	& $0.454$	& $0.608$	& $1.93$	& $6.42$	& $-297.432^{\ddagger}$	&  $\approx + 54\%$ \\\hline
    $\optrustconf{}{}{2}$ 	& $0.457$	& $0.584$	& $2.31$	& $6.08$	& $-248.968^{\ddagger}$	& $\approx + 45\%$ \\\hline
    \rowcolor{lightgray}
    $\optrustconf{}{}{3}$ 	& $0.607$	& $0.593$	& $4.92$	& $3.72$	& $-77.760^{\ddagger}$	& $\approx - 14\%$ \\\hline
    \multicolumn{7}{l}{${}^{\star}$ Median; \qquad ${}^{\dagger}$ computed as $(s^+ - s^-)/(s^+ + s^-)$; \qquad  ${}^{\ddagger} p < 0.001$.}
  \end{tabular}
  \caption{Wilcoxon signed-rank significance tests of distances derived using the geometric distance $\distance{\cdot}{\cdot}$. In grey are the cases where the proposed operator $\optrustconf{}{}{} \in \set{\optrustconf{}{}{1}, \optrustconf{}{}{2}, \optrustconf{}{}{3}, \optrustconf{}{}{n}}$ did not outperform J{\o}sang's operators. Results are ordered by the increment of performance in descending order (i.e. the first row is the best one).\label{tablestatb}}
\end{table*}

\begin{figure*}[htb]
  \centering
  \begin{subfigure}[b]{0.48\textwidth}
    \resizebox{\textwidth}{!}{\input{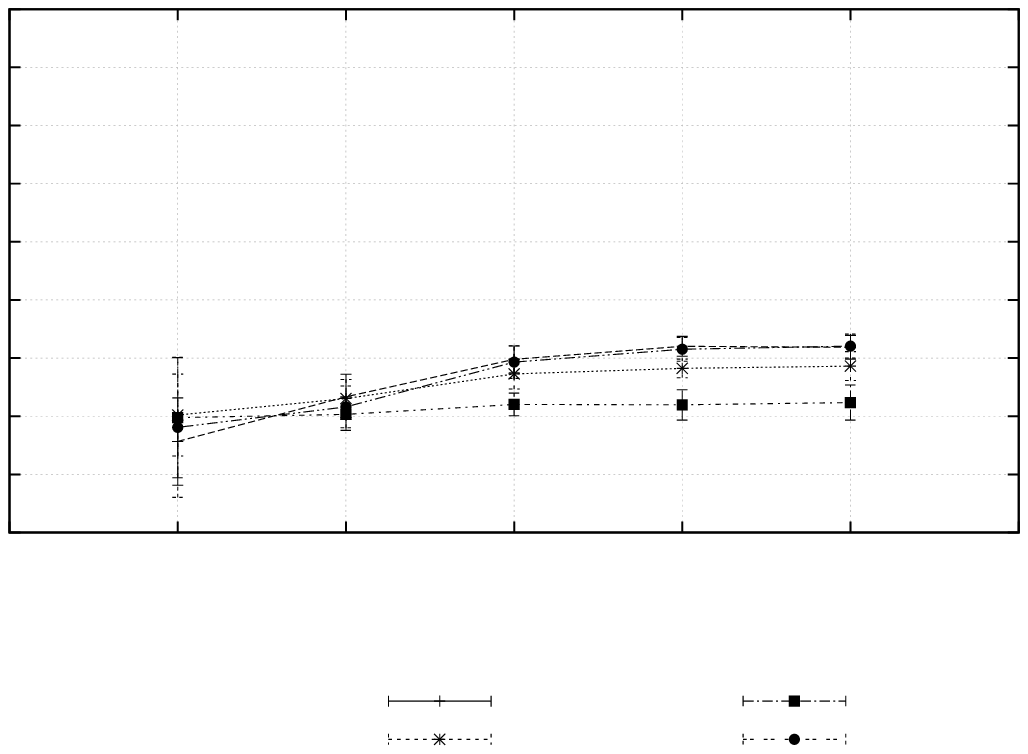}}
    \caption{}
    \label{fig:averaged-expected}
  \end{subfigure}
  \qquad
  \begin{subfigure}[b]{0.48\textwidth}
    \resizebox{\textwidth}{!}{\input{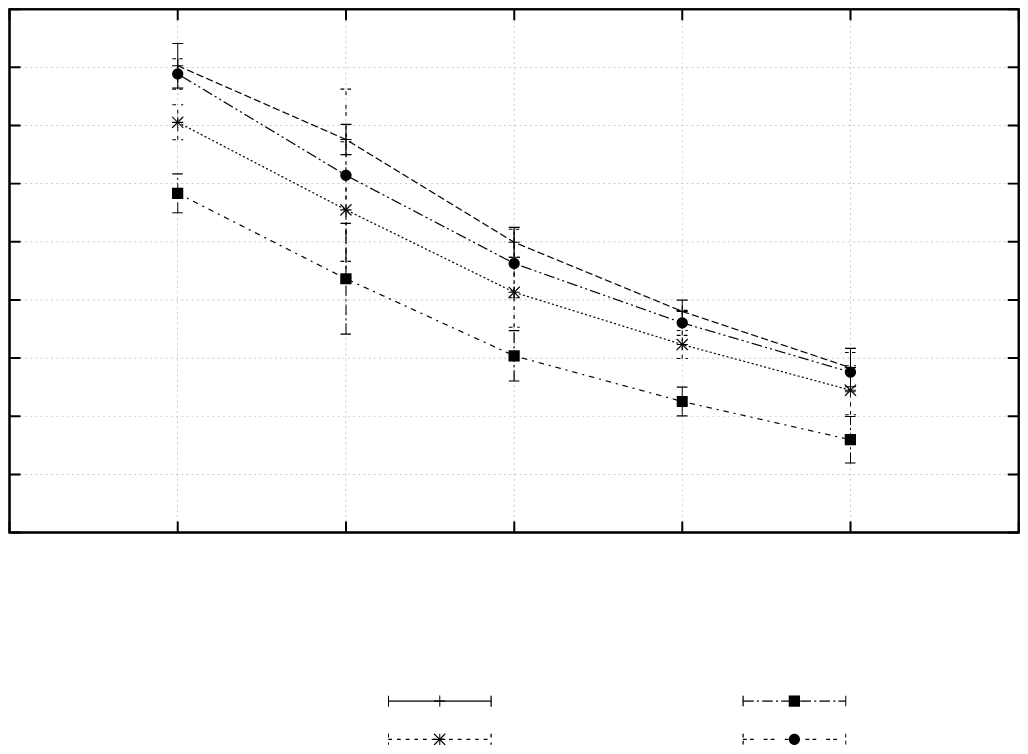}}
    \caption{}
    \label{fig:averaged-geometric}
  \end{subfigure}

  \begin{subfigure}[b]{0.48\textwidth}
    \resizebox{\textwidth}{!}{\input{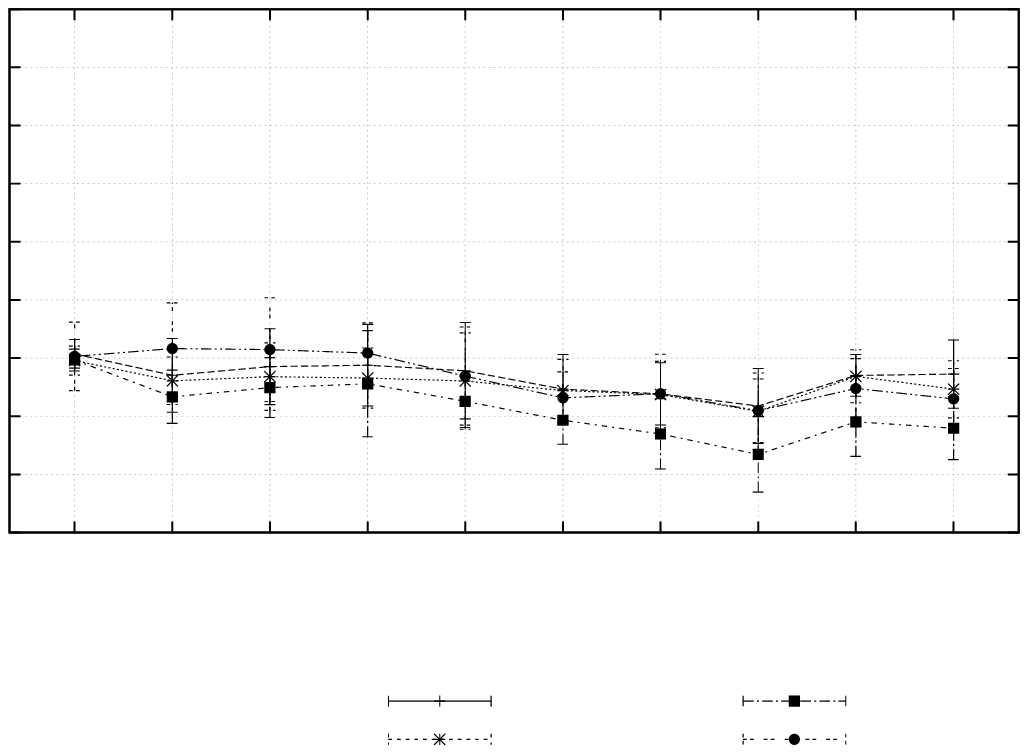}}
    \caption{}
    \label{fig:averaged-expected-bootstrap}
  \end{subfigure}
  \qquad
  \begin{subfigure}[b]{0.48\textwidth}
    \resizebox{\textwidth}{!}{\input{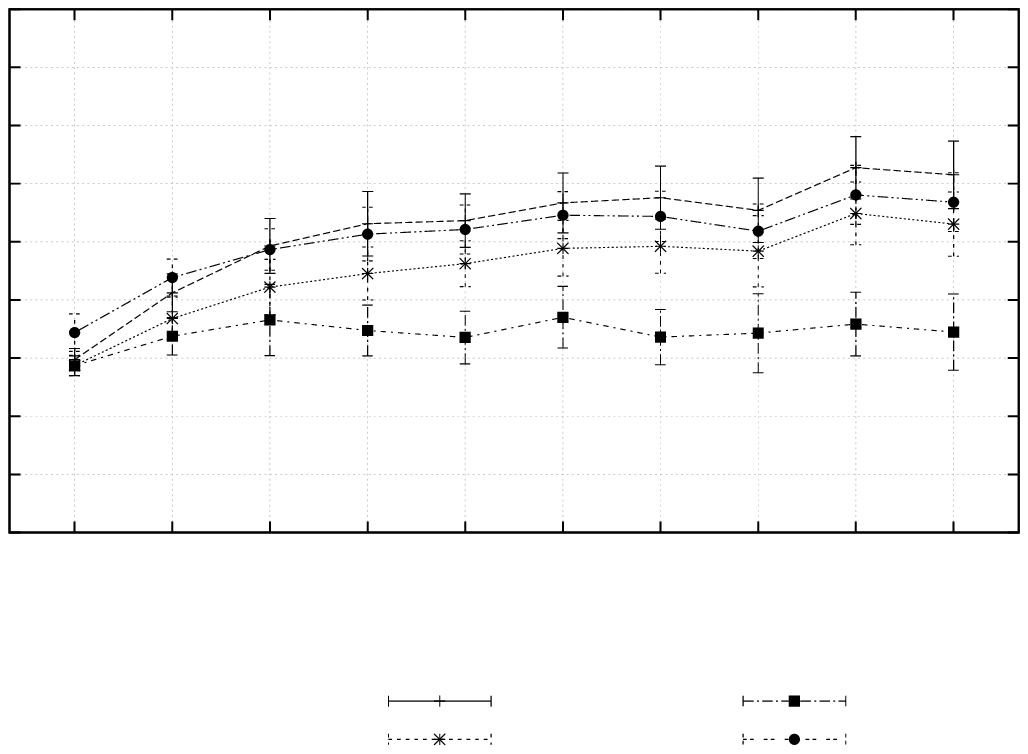}}
    \caption{}
    \label{fig:averaged-geometric-bootstrap}
  \end{subfigure}  
  \caption{{Mean and Standard Deviation of $\overline{r_E(\agent{z})}$ (resp. $\overline{r_G(\agent{z})}$) w.r.t. $\linkprob$, Fig. (\subref{fig:averaged-expected}) (resp. Fig. (\subref{fig:averaged-geometric})), and w.r.t. $\lboot$, Fig. (\subref{fig:averaged-expected-bootstrap}) (resp. Fig. (\subref{fig:averaged-geometric-bootstrap})). Please note that the scale is logarithmic.}}
  \label{fig:results}
\end{figure*}

\subsection{Analysis Using the Wilcoxon Test}
\label{sec:analys-using-wilc}

Since the Kolmogorov--Smirnov tests reported that the distribution of the differences between each pair of distance 
$\tuple{\distanceexp{{\opinionag{\theagent}{\agent{z}}}_{|\optrustconf{}{}{}}}{\opinionag{Exp}{\agent{z}}}, \distanceexp{{\opinionag{\theagent}{\agent{z}}}_{|J}}{\opinionag{Exp}{\agent{z}}}}$ 
(resp. \linebreak  $\tuple{\distance{{\opinionag{\theagent}{\agent{z}}}_{|\optrustconf{}{}{}}}{\opinionag{Exp}{\agent{z}}}, \distance{{\opinionag{\theagent}{\agent{z}}}_{|J}}{\opinionag{Exp}{\agent{z}}}}$) 
are significantly different from normal distributions ($p<0.001$), we analysed them using the  Wilcoxon Signed-Rank Test (WSRT).

This test allows us to conclude whether or not the median of the differences of such pairs of distances is  statistically equal to 0.
Moreover, looking at the median of the distribution of each distance, we can also verify the significance of the direction of the difference. In other words, if the distribution of $\distanceexp{{\opinionag{\theagent}{\agent{z}}}_{|\optrustconf{}{}{}}}{\opinionag{Exp}{\agent{z}}}$ has median equal to $a$ and the median of the distribution of $\distanceexp{{\opinionag{\theagent}{\agent{z}}}_{|J}}{\opinionag{Exp}{\agent{z}}}$ is $b$, we can verify the hypothesis that the difference is significantly positive (if $a > b$) or negative (if $a < b$).
Therefore, if the difference is significant, this tests shows that one distance from the ground truth is significantly higher than the other. 
Furthermore, WSRT calculates the sum of the ranks of the pairwise positive differences $s^+$ and negative differences  $s^-$. This can be used to indicated the size of this difference: we consider the following simple formula for determining this size which turns to be our measure of increment of performance, namely $(s^+ - s^-)/(s^+ + s^-)$ .

For improving the readability of the results, we grouped the results of the WSRT test according to the type of distance used. Table \ref{tablestata} shows the results of the WSRT considering measure computed using the expected value distance $\distanceexp{\cdot}{\cdot}$, while Table \ref{tablestatb} shows the results of the WSRT considering measure computed using the geometrical distance $\distance{\cdot}{\cdot}$.

From Tables \ref{tablestata} and \ref{tablestatb} we can conclude that, regardless of the choice of the operator and the type of measure used, the difference between the opinions determined with the proposed operators and the J{\o}sang's is significantly different ($p < 0.001$).

Concerning the expected value distances, from Table \ref{tablestata} we can see that the WSRT highlights that in the case that $\optrustconf{}{}{1}$ with $\fusionop{1}$ is used, the derived opinion is significantly ($\approx + 5 \%$) closer to the ground truth than the opinion computed using J{\o}sang's operators. This is not true according to the other choices of operators, which return opinion whose distance from the ground truth is greater (between $5\%$ and $19\%$) than J{\o}sang's operator.

However, if we consider the graphical distances, from Table \ref{tablestatb} we can see that (in order) $\optrustconf{}{}{1}$, $\optrustconf{}{}{n}$, and $\optrustconf{}{}{2}$, each of which with $\fusionop{1}$, outperform J{\o}sang's operators. Comparing these increments of performances, we can also see that the opinions derived using these operators are much closer to the ground truth (between $\approx +56\%$ and $\approx +45\%$) than J{\o}sang's operators.


\subsection{Results w.r.t. Experiment Parameters}
\label{sec:dynamics-results}

Considering the dynamics of the results,  Figure \ref{fig:results} depicts the mean and the standard deviation\footnote{Although in Sects. \ref{sec:distr-dist} and \ref{sec:analys-using-wilc} we show that the distances are not normally distributed and thus from a statistical point of view medians rather than means should be considered. Here we are more interested in the qualitative dynamics of values obtained by varying the parameters of the experiment, and thus we rely on graphical representations of mean and standard deviation.} of $\overline{r_E(\agent{z})}$ and $\overline{r_G(\agent{z})}$ for each set of operators used --- \viz{} $\optrustconf{}{}{1}$ and $\fusionop{1}$,  \optrustconf{}{}{2} and \fusionop{1}, \optrustconf{}{}{2} and \fusionop{1}, \optrustconf{}{}{n} and \fusionop{1} --- w.r.t. the two variables considered, namely the probability of connections $\linkprob$ (Figures \ref{fig:averaged-expected} and \ref{fig:averaged-geometric}), and the bootstrap time $\lboot$ (Figures \ref{fig:averaged-expected-bootstrap} and \ref{fig:averaged-geometric-bootstrap}).

Considering that distances computed using the expected value distance measure, from Figure \ref{fig:averaged-expected} we can infer that on average J{\o}sang's operators are performing better for small values of probability of connections $\linkprob$, and the greater the $\linkprob$, the better are the performance of operators $\optrustconf{}{}{1}$, $\optrustconf{}{}{n}$ and $\optrustconf{}{}{2}$ (each of which with $\fusionop{1}$). A visual inspection of Figure \ref{fig:averaged-expected-bootstrap}, however, does not highlight any specific pattern or regularity in the dynamics of the system varying $\lboot$ (considering expected value distance).

On the other hand, if we consider the results derived using the geometric distance, Fig. \ref{fig:averaged-geometric} qualitatively shows that the greater the probability of connections $\linkprob$, the more similar the operators we propose in this paper are to J{\o}sang's. In fact, the more connected the network, the more the bootstrapping phase is important, and this is independent of the choice of operators.
However, when we are considering the dynamics of the bootstrapping phase (Fig. \ref{fig:averaged-geometric-bootstrap}), we conclude that the smaller the uncertainty (i.e. the greater the number of interactions among the agents during the bootstrapping phase), the better the proposed operators perform. It is worth to notice that for $\lboot = 2$, which leads to a high uncertain opinions, $\optrustconf{}{}{1}$ and $\fusionop{1}$, $\optrustconf{}{}{2}$ and $\fusionop{1}$, $\optrustconf{}{}{3}$ and $\fusionop{1}$ perform similarly to J{\o}sang's operators, while choosing $\optrustconf{}{}{n}$ and $\fusionop{1}$ leads to a significantly better result. We will investigate this interesting results in future works.

\subsection{Summary}

To summarise our empirical evaluation, we observe that:
\begin{enumerate}
\item the operators $\optrustconf{}{}{1}$, $\optrustconf{}{}{2}$, $\optrustconf{}{}{3}$, $\optrustconf{}{}{n}$ (in conjunction to $\fusionop{1}$), similarly to J{\o}sang's operators, return opinions whose expected value distance distribution from the ground truth is close to a Gamma function (Figures \ref{fig:distribution-expected-josang}, \ref{fig:distribution-expected-traditional}, \ref{fig:distribution-expected-parallel}, \ref{fig:distribution-expected-naive});
\item the operators $\optrustconf{}{}{1}$, $\optrustconf{}{}{2}$, $\optrustconf{}{}{n}$ (in conjunction with $\fusionop{1}$), differ from J{\o}sang's operators and $\optrustconf{}{}{3}$ with $\fusionop{1}$, and return opinions whose geometric distance distribution from the ground truth shows some qualitative regularity resembling a Gamma function or a lognormal distribution (Figures \ref{fig:distribution-geometric-josang}, \ref{fig:distribution-geometric-traditional}, \ref{fig:distribution-geometric-parallel}, \ref{fig:distribution-geometric-naive});
\item the operator $\optrustconf{}{}{1}$ with $\fusionop{1}$ outperforms J{\o}sang's operators in a statistically significant manner, both considering the expected value distance ($\approx +5\%$) and the geometrical distance ($\approx + 56\%$);
\item the rank of operators (each of which used in conjunction with \fusionop{1}) w.r.t. their performances is independent from the choice of expected value distance, or geometrical distance, and is as follows: $\optrustconf{}{}{1} \succ \optrustconf{}{}{n} \succ \optrustconf{}{}{2} \succ \optrustconf{}{}{3}$;
\item the less the probability of connections, the more $\optrustconf{}{}{} \in \set{\optrustconf{}{}{1}, \optrustconf{}{}{2}, \optrustconf{}{}{3}, \optrustconf{}{}{n}}$ returns opinions closer (according to the graphical distance) to the ground truth than J{\o}sang's operators;
\item the less the uncertainty (i.e. the more the bootstrap time), the more $\optrustconf{}{}{} \in \set{\optrustconf{}{}{1}, \optrustconf{}{}{2}, \optrustconf{}{}{3}, \optrustconf{}{}{n}}$ returns opinions closer (according to the graphical distance) to the ground truth than J{\o}sang's operators.
\end{enumerate}

\section{Conclusions and Future Works}
\label{sec:conclusions}

The discount and the fusion operators play an important role in standard Subjective Logic, and form the core of the Beta Reputation System. In fact, they are used to combine and discount reputation information from multiple agents within a trust network.

In this paper, following our earlier work in \cite{Cerutti2013a,Cerutti2013}, we introduced a set of intuitive desiderata that operators for discounting and fusion of opinions should provide. From these, we derived a set of requirements and a family of operators, and proved that these satisfy the desiderata, while J{\o}sang's operators do not . We empirically evaluated the derived operators in a trust scenario and the results shown in Section \ref{sec:results} suggest that: 
\begin{itemize}
\item one operator taken from the family satisfying the desiderata always outperforms J{\o}sang's operators;
\item according to the geometrical distance among opinions, most of the operators satisfying the desiderata outperform J{\o}sang's operators;
\item there are relationships between the structure of the trust network and the achieved increments of performances.
\end{itemize}

In particular, the Wilcoxon signed-rank significance test discussed in Section \ref{sec:analys-using-wilc}, shows that the discounting operator ($\optrustconf{}{}{1}$), used in conjunction with the fusion operator \fusionop{1}, returns opinions closer to the ground truth than J{\o}sang's operators of $5\%$ considering the expected value distance, and of $56\%$ considering the graphical distance. Therefore, it seems that allowing a reduction of the amount of uncertainty in discounting opinion results on an increment of the performances not only geometrically, but also when the expected values are considered. 
 
An empirical evaluation of the graphical operators on real cases, e.g. \citep{Guha2004}, is already envisaged as the main future work.
In addition, we want to develop graphical operators analogous to other Subjective Logic operators, and we intend to study these, as well as investigate their properties. 


\bibliographystyle{spbasic}
\bibliography{biblio}

\appendix

\section{The Geometry of Subjective Logic}
\label{sec:geom-subj-logic}

A SL opinion $O \triangleq \opinion{O}$ is a point in the $\realset^3$ space, identified by the coordinate \belief{O} for the first axis, \disbelief{O} for the second axis, and \uncertainty{O} for the third axis. However, due to the requirement that $\belief{O} + \disbelief{O} + \uncertainty{O} = 1$, an opinion is a point inside (or at least on the edges of) the triangle $\tri{BDU}$ shown in Fig. \ref{fig:3d-sl}, where $B = \tuple{1, 0, 0}, D=\tuple{0,1,0}, U=\tuple{0,0,1}$.

\begin{figure}[h]
  \centering
  \includegraphics[scale=0.6]{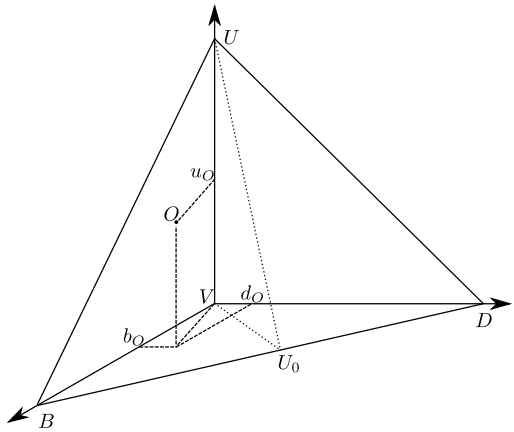}
  \caption{The Subjective Logic plane region}
  \label{fig:3d-sl}
\end{figure}

\begin{definition}
  The \emph{Subjective Logic plane region \tri{BDU}} is the triangle whose vertices are the points $B \triangleq \tuple{1,0,0}$, $D\triangleq\tuple{0,1,0}$, and $U\triangleq\tuple{0,0,1}$ on a $\realset^3$ space where the axes are respectively the one of belief, disbelief, and uncertainty predicted by SL.
\end{definition}

Since an opinion is a point inside  triangle $\tri{BDU}$, it can be mapped to a point in Fig. \ref{fig:basic-sl}. This representation is similar to the one used in \citep{Josang2001} for representing opinions in SL, but here the belief and disbelief axes are swapped.

\begin{figure}[h]
  \centering
  \includegraphics[scale=0.6]{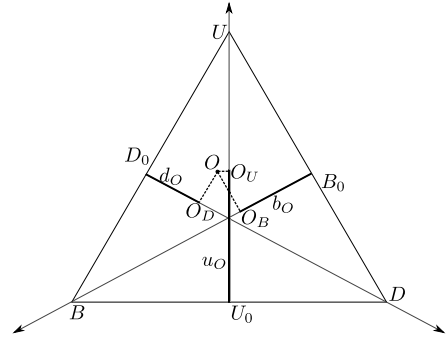}
  \caption{An opinion $O \triangleq \opinion{O}$ in SL after the $1:\frac{\sqrt{3}}{\sqrt{2}}$ scale. The belief axis is the line from $B_0$ (its origin) toward the $B$ vertex, the disbelief axis is the line from $D_0$ toward the $D$ vertex, and the uncertainty axis is the line from $U_0$ toward the $U$ vertex}
  \label{fig:basic-sl}
\end{figure}

In order to keep the discussion consistent with J{\o}sang's work \citep{Josang2001}, in what follows we will scale triangle \tri{BDU} by a factor $1:\frac{\sqrt{3}}{\sqrt{2}}$ thus obtaining that $|\vv{B_0 B}| = |\vv{D_0 D}| = |\vv{U_0 U}| = 1$. 

These geometric relations lie at the heart  of the Cartesian transformation operator which is the subject of the next subsection.

\subsection{The Cartesian Representation of Opinions}
\label{sec:cart-repr-opin}

As shown in \ref{sec:geom-subj-logic}, an opinion in SL can be represented as a point in a planar figure (Fig. \ref{fig:basic-sl}) laying on a Cartesian plane. In this section we will introduce the Cartesian transformation operator which returns the Cartesian coordinate of an opinion.

First of all, let us define the axes of the Cartesian system we will adopt. 

\begin{definition}
  Given the SL plane region \tri{BDU}, the \emph{associated Cartesian system} is composed by two axes, named respectively $x, y$, where the unit vector of the $x$ axis $\versor{x} = \frac{1}{|\vv{BD}|} \vv{BD}$, the unit vector of the $y$ axis $\versor{y} = \versor{\uncertainty{}}$, and $B$ is the origin.
\end{definition}

Figure \ref{fig:cartesian-plane} depicts this Cartesian system.

\begin{figure}[h]
  \centering
  \includegraphics[scale=0.5]{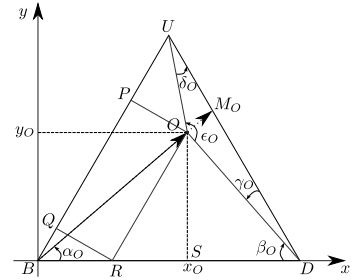}
  \caption{An opinion and its representation in the Cartesian system}
  \label{fig:cartesian-plane}
\end{figure}

The correspondence between the three values of an opinion and the corresponding coordinate in the Cartesian system we defined is shown in the following proposition (proved in \citep{Cerutti2013}).

\begin{restatable}{propn}{tocartesian} \emph{\citep[Prop. 1]{Cerutti2013}}
\label{propn:cartesian-transformation}
  Given a SL plane region \tri{BDU} and its associated Cartesian system $\tuple{x, y}$, an opinion $O \triangleq \opinion{O}$ is identified by the coordinate $\tuple{x_O, y_O}$ \suchthat:
  \begin{itemize}
  \item $\displaystyle{x_O \triangleq \frac{\disbelief{O} + \uncertainty{O}~\cos(\frac{\pi}{3})}{\sin(\frac{\pi}{3})}}$
  \item $y_O \triangleq \uncertainty{O}$
  \end{itemize}
\end{restatable}

\begin{proof}
    Proving that $y_O \triangleq \uncertainty{O}$ is trivial. 

    Let us focus on the first part of the proposition. Consider Figure \ref{fig:cartesian-plane}. Given $O$, we note that the for the point $P$,  $\frac{1}{|\vv{P O}|} \vv{P O} = \versor{b}$ (\ie{} $\vv{PO}$ is parallel to the disbelief axis) and $\frac{1}{|\vv{BP}|} \vv{BP} = \frac{1}{|\vv{BU}|} \vv{BU}$ (\ie{} $P$ is on the line $\vv{BP}$), and therefore $\ang{BPO} = \frac{\pi}{2}$. Then we must determine $Q$ and $R$ \suchthat{} $\vv{Q R} = \vv{P O}$ and $y_{R} = 0$. By construction $|\vv{P O}| = |\vv{Q R}| = \disbelief{O}$, $\ang{Q R B} = \frac{\pi}{6}$, $\ang{ORD} = \frac{\pi}{3}$, and $x_O \triangleq |\vv{BS}| =  |\vv{BR}| + |\vv{RS}|$, where $|\vv{BR}| = \frac{\disbelief{O}}{\sin(\frac{\pi}{3})}$, and $|\vv{RS}| = \frac{\uncertainty{O}}{\sin(\frac{\pi}{3})} \cos(\frac{\pi}{3})$. 
\end{proof}

There are some notable elements of Fig. \ref{fig:cartesian-plane} that we will repeatedly use below, and we therefore define them as follows:
\begin{itemize}
\item the angle $\alpha_O$ determined by the $x$ axis and the vector $\vv{BO}$;
\item the three angles ($\gamma_O, \delta_O,$ and $\epsilon_O$) of the triangle $\tri{ODU}$, namely the triangle determined by linking the point $O$ with the vertex $D$ and $U$ through straight lines.
\end{itemize}

\begin{definition}
  \label{def:properties}
  Given the SL plane region \tri{BDU}, given $O=\opinion{O}$ whose coordinates are $\tuple{x_O, y_O}$ where $\displaystyle{x_O \triangleq \frac{\disbelief{O} + \uncertainty{O}~\cos(\frac{\pi}{3})}{\sin(\frac{\pi}{3})}}$ and $y_O \triangleq \uncertainty{O}$, let us define and (via trivial trigonometric relations) compute the following.
  \begin{itemize}
  \item $\displaystyle{\alpha_O \triangleq \ang{OBD}} = $

    $\left\{
        \begin{array}{l l}
          0 & \mbox{if } \belief{O} = 1\\
          \displaystyle{\arctan\left(\frac{\uncertainty{O}~\sin(\frac{\pi}{3})}{\disbelief{O} + \uncertainty{O}~\cos(\frac{\pi}{3})}\right)} & \mbox{otherwise}\\
        \end{array} \right.$;
  \item $\displaystyle{\beta_O \triangleq \ang{ODB} } = $

    $\left\{
      \begin{array}{l l}
        \displaystyle{\frac{\pi}{3}} & \mbox{if } \disbelief{O} = 1\\
        \displaystyle{\arctan\left(\frac{\uncertainty{O}~\sin(\frac{\pi}{3})}{1-(\disbelief{O}+\uncertainty{O}~\cos(\frac{\pi}{3}))}\right)} & \mbox{otherwise}\\
      \end{array}\right.$;
  \item $\displaystyle{\gamma_O \triangleq \ang{ODU} = \frac{\pi}{3}} - \beta_O$;
  \item $\displaystyle{\delta_O \triangleq \ang{OUD} = }$

      $\left\{
        \begin{array}{l l}
          0 & \mbox{if } \uncertainty{O} = 1\\
          \displaystyle{\arcsin\left(\frac{\belief{O}}{|\vv{OU}|}\right)} & \mbox{otherwise}\\
        \end{array}\right.$;
  \item $\epsilon_O \triangleq \ang{DOU} = \pi - \gamma_O - \delta_O$;
  \end{itemize}

\noindent
where $\displaystyle{|\vv{OU}| = \sqrt{\frac{1}{3} (1 + \disbelief{O} - \uncertainty{O})^2 + \belief{O}^2}}$.

The angle $\alpha_O$ is called the  \emph{direction of $O$}.

Equivalently, we can  write $\vv{BO}$ or $\tuple{B, \alpha_O, |\vv{BO}|}$.

\end{definition}

Finally, as an element of SL is bounded to have its three components between $0$ and $1$, we are also interested in determining the point $M_O$ such that the vector $\vv{BM_O}$ has the maximum magnitude given (a) the direction $\alpha_O$ of an opinion $O$, and (b) $M_O$ is a SL opinion. In other words, determining the magnitude of $\vv{BM_O}$ will allow us to re-define the vector $\vv{BO}$ as a fraction  of $\vv{BM_O}$. 

\begin{definition}
  \label{defn:max-vector}
  Given the SL plane region \tri{BDU}, and $O \triangleq \opinion{O}$ whose coordinates are $\tuple{x_O, y_O}$ where $\displaystyle{x_O \triangleq \frac{\disbelief{O} + \uncertainty{O}~\cos(\frac{\pi}{3})}{\sin(\frac{\pi}{3})}}$ and $y_O \triangleq \uncertainty{O}$, and $\displaystyle{\alpha_O \triangleq \ang{OBD} = \arctan\left(\frac{\uncertainty{O}~\sin(\frac{\pi}{3})}{\disbelief{O} + \uncertainty{O}~\cos(\frac{\pi}{3})}\right)}$, let us define $M_O \triangleq \tuple{x_{M_O}, y_{M_O}}$ as the intersection of the straight line passing for $O$ and $B$, and the straight line passing for $U$ and $D$, and thus define the following.
  \begin{itemize}
  \item $\displaystyle{x_{M_O} \triangleq \frac{2 - y_O + \tan(\alpha_O)~ x_O}{\tan(\alpha_O) + \sqrt{3}}}$;

  \item $\displaystyle{y_{M_O} \triangleq -\sqrt{3}~ x_{M_O} + 2}$.
  \end{itemize}
\end{definition}

\section{Proofs}
\label{sec:proofs}

\propnaive*
\begin{proof}\mbox{}

  \begin{itemize}
  \item[$i$.] Let us prove that $0 \leq \belief{W} \leq 1$, $0 \leq \disbelief{W} \leq 1$,  $0 \leq \uncertainty{W} \leq 1$.

    To prove that $\belief{\W{}} \geq 0$, $\disbelief{\W{}} \geq 0$, and $\uncertainty{\W{}} \geq 0$ is trivial since $\C{}$ and $\T{}$ are opinions.

    $\belief{\W{}} = \belief{\C{}} \cdot \belief{\T{}} \leq 1$ is immediate since $\C{}$ and $\T{}$ are opinions.

    $\disbelief{\W{}} = \belief{} \cdot \disbelief{\T{}} + \disbelief{\C{}} \leq 1$ can be rewritten as $\disbelief{\T{}} \leq 1 + \frac{\uncertainty{\C{}}}{\belief{\C{}}}$ if $\belief{\C{}} \neq 0$, or $\disbelief{\C{}} \leq 1$ otherwise. Both in-equations are verified since $\C{}$ and $\T{}$ are opinions.

    $\uncertainty{\W{}} = \belief{} \cdot \uncertainty{\T{}} + \uncertainty{\C{}} \leq 1$ can be rewritten as $\disbelief{\T{}} \leq 1 + \frac{\disbelief{\C{}}}{\belief{\C{}}}$ if $\belief{\C{}} \neq 0$, or $\uncertainty{\C{}} \leq 1$ otherwise. Both in-equations are verified since $\C{}$ and $\T{}$ are opinions.
    
    Finally, $\belief{\W{}} + \disbelief{\W{}} + \uncertainty{\W{}} = \belief{\C{}} (\belief{\T{}} + \disbelief{\T{}} + \uncertainty{\T{}}) + \disbelief{\C{}} + \uncertainty{\C{}} = 1$

  \item[$ii$.] Given $\C{} = \tuple{1, 0, 0}$, $\W{} = \optrustconf{\T{}}{\C{}}{n}$ is such that:
    \begin{itemize}
    \item $\belief{\W{}} = \belief{\C{}} \cdot \belief{\T{}} = \belief{\T{}}$;
    \item $\disbelief{\W{}} = \belief{\C{}} \cdot \disbelief{\T{}} + \disbelief{\C{}} = \disbelief{\T{}}$;
    \item $\uncertainty{\W{}} = \belief{\C{}} \cdot \uncertainty{\T{}} + \uncertainty{\C{}} = \uncertainty{\T{}}$.
    \end{itemize}

  \item[$iii$.] Given $\C{} = \tuple{0, 0, 1}$, $\W{} = \optrustconf{\T{}}{\C{}}{n}$ is such that:
    \begin{itemize}
    \item $\belief{\W{}} = \belief{\C{}} \cdot \belief{\T{}} = 0 = \belief{\C{}}$;
    \item $\disbelief{\W{}} = \belief{\C{}} \cdot \disbelief{\T{}} + \disbelief{\C{}} = \disbelief{\C{}}$;
    \item $\uncertainty{\W{}} = \belief{\C{}} \cdot \uncertainty{\T{}} + \uncertainty{\C{}} = \uncertainty{\C{}}$.
    \end{itemize}

  \item[$iv$.] By contradiction, $\belief{\W{}} = \belief{\C{}} \cdot \belief{\T{}} > \belief{\T{}}$ leads to $\belief{\C{}} > 1$, which is impossible.\qed
  \end{itemize}
\end{proof}

\propadmissiblespace*
\begin{proof}
  By Definition \ref{defn:admissiblespace}, $\admissibleopinionset{\T{}{}} = \set{X \in \opinionset{} | \belief{X} \leq \belief{\T{}{}}}$. From Prop. \ref{propn:cartesian-transformation},
  \[
  \begin{split}
    \belief{X} & \leq \belief{\T{}}\\
    y_X & \geq \sqrt{3}x_X + 2(1 - \belief{\T{}})\\
    \uncertainty{X} & \geq 1 - \disbelief{X} - \belief{\T{}}
  \end{split}
  \]

  Therefore:
  \begin{itemize}
  \item if $\disbelief{X} = 0$, $\uncertainty{X} \geq 1 - \belief{\T{}}$ (limit case $\tuple{\belief{\T{}}, 0, 1 - \belief{\T{}}}$;
  \item if $\uncertainty{X} = 0$, $\disbelief{X} \geq 1 - \belief{\T{}}$ (limit case $\tuple{\belief{\T{}}, 1-\belief{\T{}}, 0} = Q$).   \qed
  \end{itemize}
\end{proof}

\thmfamilydiscount*

\begin{proof}
    Proving the thesis in the limit case is trivial. In the following we will assume, without loss of generality, that $\alpha_{C'} \neq \frac{\pi}{2}$, $\alpha_{C'} \neq -\frac{\pi}{3}$, $\alpha_{C'} \neq \frac{2}{3} \pi$.

    \emph{(i.)} $W = \opinion{W}$ must respect 
    \begin{equation}
      \label{eq:1}
      \uncertainty{W} + \disbelief{W} \leq 1
    \end{equation}

    From Def. \ref{defn:combination} it is clear that Equation \ref{eq:1} can be rewritten as follows.

    \begin{equation}
      \label{eq:2}
      \begin{split}
        & \uncertainty{T} + \disbelief{T} + \\ & + \frac{r_C}{2} \frac{2 \sqrt{\tan^2(\alpha_{C'}) + 1}}{|\tan(\alpha_{C'}) + \sqrt{3}|} ~\belief{T} ~(\sin(\alpha_{C'}) + \sqrt{3}\cos(\alpha_{C'})) \leq 1
      \end{split}
    \end{equation}
    
    In turn, using the relation $\tan(\alpha_{C'}) = \frac{\sin(\alpha_{C'})}{\cos(\alpha_{C'})}$, this can be rewritten as 

    \[
    \uncertainty{T} + \disbelief{T} + r_C \belief{T} \leq 1
    \]

    \noindent
    which entails the requirement that $r_C \leq \frac{1 - \uncertainty{T} - \disbelief{T}}{\belief{T}} = \frac{\belief{T}}{\belief{T}} = 1$. However, from definition \ref{defn:combination}, we know that $r_{C} \leq 1$, fulfilling this requirement.

    \emph{(ii.)} $C = \tuple{1, 0, 0}$ implies that $|\vv{BC}| = 0$ and thus $r_C = 0$. Therefore, from Def. \ref{defn:combination}, $\uncertainty{W} = \uncertainty{T} + \sin(\alpha_{C'}) r_C |\vv{TM_{C'}}| = \uncertainty{T}$ and this results also implies that $\disbelief{W} = \disbelief{T}$. Since Point 1 shows that $W$ is an opinion in SL, we conclude that $W = T$.

    \emph{(iii.)} $C = \tuple{0, 0, 1}$ implies $r_C = 1$, $\frac{\frac{\pi}{3} \epsilon_T}{\frac{\pi}{3}} - \beta_T \leq \alpha_{C'} \leq \epsilon_T - \beta_T$, and thus $\alpha_{C'} = \epsilon_T - \beta_T = \epsilon_T - \beta_T = \frac{2}{3} \pi - \delta_t$. Therefore, we obtain that $\uncertainty{W} = \uncertainty{T} + \frac{\belief{T}}{2} (1 + \frac{\sqrt{3}}{\tan(\delta_T)})$. From Definition \ref{def:properties} and the trigonometric property that $\tan(\arcsin(v)) = \frac{v}{\sqrt{1 - v^2}}$ we obtain that $\uncertainty{W} = \uncertainty{T} + \frac{\belief{T}}{2} + \frac{\sqrt{3}}{2} \sqrt{|\vv{TU}|^2 - \belief{T}^2}$. From Definition \ref{def:properties} we can write:

    \begin{equation}
      \label{eq:5}
      \begin{split}
        \uncertainty{W} & = \uncertainty{T} + \frac{\belief{T}}{2} + \frac{\sqrt{3}}{2} \frac{1 + \disbelief{T} - \uncertainty{T}}{\sqrt{3}}\\
        & = \frac{1}{2} (1 + \belief{T} + \disbelief{T} + \uncertainty{T} ) = 1
      \end{split}
    \end{equation}

    Similarly, $\disbelief{W} = \disbelief{T} + \frac{\uncertainty{T}}{2} - \frac{1}{2} + \frac{\sqrt{3}}{2} \frac{1}{\sin(\delta_T)} \belief{T} = \disbelief{T} + \frac{\uncertainty{T}}{2} - \frac{1}{2} + \frac{\sqrt{3}}{2} |\vv{TU}|$. From Def. \ref{def:properties} we have

    \begin{equation}
      \label{eq:6}
      \begin{split}
        \disbelief{W} & = \disbelief{T} + \frac{\uncertainty{T}-1}{2} + \frac{3}{4}\belief{T} - \frac{\sqrt{3}}{4} \belief{T} \frac{1 + \disbelief{T} - \uncertainty{T}}{\sqrt{3} \belief{T}} \\ 
        & = \frac{1}{4} (4\disbelief{T} + 2 \uncertainty{T} - 2 + 3 \belief{T} - 1 + \uncertainty{T} - \disbelief{T}) = 0
      \end{split}
    \end{equation}

    From Equations \ref{eq:5} and \ref{eq:6}, together with Point 1, it follows that $W = \tuple{0,0,1}=C$.

    \emph{(iv.)} Suppose instead $\belief{W} > \belief{T}$.

    \begin{equation*}
      \begin{split}
        & 1 - \disbelief{W} - \uncertainty{W} > 1 - \disbelief{T} - \uncertainty{T}\\
        & \disbelief{W} + \uncertainty{W} < \disbelief{T} + \uncertainty{T} \\
        & \disbelief{T} + \sin(\alpha_{C'} + \frac{\pi}{3}) \frac{r_C}{\sin(\alpha_{C'} + \frac{\pi}{3})} + \uncertainty{T} < \disbelief{T} + \uncertainty{T}\\
        & r_C < 0
      \end{split}
    \end{equation*}
    
    \noindent
    but $0 \leq r_C \leq 1$. \emph{Quod est absurdum.} \qed
  \end{proof}

\fusionprop*
  \begin{proof}
    \emph{($i.$)} To prove that $\opinion{\fusionop{1}(W_1, \ldots, W_n)}$ is an opinion, we have to show that 
$\uncertainty{\fusionop{1}(W_1, \ldots, W_n)} + \disbelief{\fusionop{1}(W_1, \ldots, W_n)} \leq 1$ holds.

$\begin{array}{rl} 
  \uncertainty{\fusionop{1}(W_1, \ldots, W_n)} + \disbelief{\fusionop{1}(W_1, \ldots, W_n)} = & \displaystyle{\frac{1}{\sum_{i=1}^n K_i} \left( \sum_{i=1}^n K_i( \uncertainty{W_i} + \disbelief{W_i})\right)}\\
  = & \displaystyle{\frac{1}{\sum_{i=1}^n K_i} \left( \sum_{i=1}^n K_i(1 - \belief{W_i})\right)}\\
  = & 1 - \displaystyle{\frac{1}{\sum_{i=1}^n K_i} \left( \sum_{i=1}^n K_i~\belief{W_i} \right)}
\end{array}$

    \emph{($ii.$)} From Prop. \ref{propn:cartesian-transformation}, 


\noindent
$\begin{cases} 
  \displaystyle{x_{\fusionop{1}(W_1, \ldots, W_n)} =  \frac{\disbelief{\fusionop{1}(W_1, \ldots, W_n)}}{\sin(\frac{\pi}{3})} + \frac{1}{2~\sin(\frac{\pi}{3})~\sum_{i=1}^n K_i} \left( \sum_{i=1}^n K_i~ \uncertainty{W_i}\right)}\\
  \displaystyle{x_{\fusionop{1}(W_1, \ldots, W_n)} = \frac{1}{\sin(\frac{\pi}{3}) \sum_{i=1}^n K_i} \left( \sum_{i=1}^n K_i~ (\disbelief{W_i} + \frac{\uncertainty{W_i}}{2}) \right)}
\end{cases}$

\noindent
Thus we obtain:


\noindent
$\begin{array}{rl} 
\disbelief{\fusionop{1}(W_1, \ldots, W_n)} = & \displaystyle{\sin(\frac{\pi}{3})\left(\frac{1}{\sin(\frac{\pi}{3}) \sum_{i=1}^n K_i} \left( \sum_{i=1}^n K_i~ (\disbelief{W_i} + \frac{\uncertainty{W_i}}{2}) \right) + \right. }\\
    & ~~~~ \displaystyle{\left.- \frac{1}{2~\sin(\frac{\pi}{3})~\sum_{i=1}^n K_i} \left( \sum_{i=1}^n K_i~ \uncertainty{W_i}\right) \right)}\\
    = & \displaystyle{\frac{1}{\sum_{i=1}^n K_i} \left( \left(\sum_{i=1}^n K_i~ (\disbelief{W_i} + \frac{\uncertainty{W_i}}{2})\right) - \left( \sum_{i=1}^n \frac{\uncertainty{W_i}}{2} \right) \right)}\\
    = & \displaystyle{\frac{1}{\sum_{i=1}^n K_i} \left( \sum_{i=1}^n K_i~ \disbelief{W_i}\right)}
\end{array}$

Since $\displaystyle{\frac{1}{\sum_{i=1}^n K_i} \left( \sum_{i=1}^n K_i~\belief{W_i} \right)} \geq 0$, then $\uncertainty{\fusionop{1}(W_1, \ldots, W_n)} + \disbelief{\fusionop{1}(W_1, \ldots, W_n)} \leq 1$ holds.

    \emph{($iii.$)} Immediate from Prop. \ref{propn:cartesian-transformation}.\qed
  \end{proof}

\end{document}